\documentclass[preprint,5p, a4paper, sort&compress,authoryear,hyperref]{elsarticle}

\usepackage[T1]{fontenc}
\usepackage[fleqn]{amsmath}
\usepackage{amssymb, amsfonts}
\usepackage{booktabs} 
\usepackage{mathtools}
\usepackage{commath}
\usepackage{amsthm}
\usepackage{mathdots}
\usepackage{microtype}
\usepackage[nolist]{acronym}
\usepackage{siunitx}	
\usepackage{booktabs}
\usepackage{bm}
\usepackage{dsfont}

\usepackage{inputenc}

\usepackage{tikz}
\usepackage{pgfplots}
\pgfplotsset{every axis/.append style={line width=1pt}}
\pgfplotsset{compat=newest}
\usetikzlibrary{shapes, fit, decorations.markings, matrix, plotmarks, positioning, fit, spy, patterns, shadows, calc, backgrounds, arrows.meta,spy,decorations.fractals}
\usepackage[font=footnotesize,caption=false]{subfig}

\newcounter{subeqn} %

\usepackage{amsthm}
\newtheorem{theorem}{Theorem}
\newtheorem{corollary}{Corollary}
\newtheorem{lemma}{Lemma}
\newtheorem{proposition}{Proposition}
\newtheorem{definition}{Definition}
\newtheorem{remark}{Remark}
\newtheorem{assumption}{Assumption}
\newtheorem{example}{Example}

\DeclarePairedDelimiter{\diagpars}{(}{)}
\newcommand{\diag}{\operatorname{diag}\diagpars}

\begin{acronym}
	\acro{MPC}{Model Predictive Control}
	\acro{DMPC}{Distributed Model Predictive Control}
	\acro{LSS}{Large-Scale System}
	\acro{mRPI}{minimal Robust Positively Invariant}
	\acro{QP}{Quadratic Programming}
	\acro{RCI}{Robust Control Invariant}
	\acro{RPI}{Robust Positive Invariant}
	\acro{PnP}{Plug and Play}
	\acro{OCP}{Optimal Control Problem}
	\acro{MPC}{Model Predictive Control}
	\acro{PV}{Photo-voltaic}
	\acro{DG}{Distributed Generation}
	\acro{DS}{Distributed Storage}
	\acro{BIC}{Bounded Integral Controller}
	\acro{VCS}{Voltage controlled source}
	\acro{MG}{Micro-Grid}
	\acro{OPF}{Optimal Power Flow}
	\acro{SoC}{State of Charge}
	\acro{OCP}{Optimal Control Problem}
	\acro{CPL}{Constant Power Load}	
	\acro{ISS}{Input-to-state Stable}		
\end{acronym}

\newcommand{\ts}[1]{{\textnormal{#1}}}

\newcommand{\pdiff}[2]{{\frac{\partial #1}{\partial #2}}}
\newcommand{\ie}{\emph{i.e.},\ }
\newcommand{\eg}{\emph{e.g.}}

\newcommand{\Nset}{\mathbb{N}}
\newcommand{\Rset}{\mathbb{R}}

\newcommand{\mc}{\mathcal}
\newcommand{\mb}{\mathbf}
\newcommand{\mbb}{\mathbb}

\DeclareMathAlphabet\mathbfcal{OMS}{cmsy}{b}{n}

\allowdisplaybreaks
\begin{document}
\begin{frontmatter}
   \title{Two-layer nonlinear control of DC-DC buck converters with meshed network topology}
   \author[Paestum]{Pablo R. Baldivieso-Monasterios}\ead{p.baldivieso@sheffield.ac.uk}
   \author[Paestum]{Mahdieh Sadabadi}\ead{m.sadabadi@sheffield.ac.uk}
   \author[Paestum,Paestum1]{George C. Konstantopoulos}\ead{g.konstantopoulos@sheffield.ac.uk}
   \address[Paestum]{Department of Automatic Control \& Systems Engineering, University of Sheffield, Mappin Street, Sheffield S1 3JD, UK}
   \address[Paestum1]{Department of Electrical and Computer Engineering, University of Patras, Patras 26500, Greece}
   \cortext[c1]{Corresponding author}
   \tnotetext[t1]{Work supported by EPSRC under grants No $EP/S001107/1$ and $EP/S031863/1$.}
  \begin{abstract}
    In this paper, we analyse the behaviour of a buck converter network that contains arbitrary, up to mild regularity assumptions, loads. Our analysis of the network begins with the study of the current dynamics; we propose a novel Lyapunov function for the current in closed-loop with a bounded integrator. We leverage on these results to analyse the interaction properties between voltages and bounded currents as well as between node voltages and to propose a two-layer optimal controller that keeps network voltages within a compact neighbourhood of the nominal operational voltage. We analyse the stability of the closed loop system in two ways: one considering the interconnection properties which yields a weaker ISS type property and a second that contemplates the network in closed loop with a distributed optimal controller. For the latter, we propose a novel distributed way of controlling a Laplacian network using neighbouring information which results in asymptotic stability. We demonstrate our results in a meshed topology network containing 6 power converters, each converter feeding an individual constant power load with values chaning arbitrarily within a pre-specified range. 
  \end{abstract}

\end{frontmatter}

\section{Introduction}
\label{sec:introduction}

The presence of DC networks has become an important feature in modern complex systems such as electric vehicles, data centers, aircrafts, spacecrafts, telecom systems, submarines, and many more which include devices such as storage units, photovoltaic panels, or electronic loads, operate in a DC setting \citep{Elsayed2015}. In addition, DC networks provide additional flexibility to handle high voltage transmission capabilities or small localised solutions in the form of a DC \ac{MG} \citep{ZambronideSouza2019}. Therefore, a DC \ac{MG} can grow from simple photovoltaic and storage implementations to highly complex and constantly evolving systems with meshed topology with an increasing tendency to decentralisation to increase reliability and power quality for local users \citep{Planas2015,Guerrero2011}.

Traditionally, control problems in DC \ac{MG}, as mentioned in \cite{Tucci2018}, involve voltage stabilisation, current or load sharing, and voltage balancing, \ie voltage at load buses operating around a nominal voltage level. As the size of DC \ac{MG} grows, so does the need of non-centralised control techniques that are capable of attaining control objectives without complete redesigns when the network changes. Several decentralised voltage control strategies for a network of DC-DC buck converters feeding different types of loads such as pure resistive loads, constant power loads, nonlinear ZIP (constant impedance, constant current, and constant power loads), and time-varying loads have been proposed in the literature. The proposed methods include robust control approaches \cite{Sadabadi2018}, plug-and-play methods \cite{Tucci2016,Tucci2018,Sadabadi2020}, passivity-based techniques \cite{NAHATA2020}, $\mathcal{L}_2$ gain-based loop shaping methods \cite{Sadabadi2021a}, output regulation approaches \cite{Silani2021}, and optimisation-based controller methods using control barrier functions \cite{Kosaraju22}. Although these methods ensure the voltage regulation and guarantee the stability of overall DC networks, they do not consider current regulation/sharing in the network of buck converters. Distributed control approaches for average voltage regulation and proportional current sharing in DC microgrids with DC-DC buck converters have been proposed in \cite{Tucci2018b, Cucuzzella2019,Trip2019,Sadabadi2021b,Nahata2020b}. In \cite{Guo2019a,Guo2019b}, overcurrent protection approaches for a single DC-DC buck converter have been proposed. However, the extension of these methods to a network of DC-DC buck converters is not straightforward. On the other hand, \cite{DePersis2018a} propose a power sharing controller which employs nonlinear consensus to obtain effective load sharing while keeping voltages bounded to a compact set. The problem of power sharing has been discussed in \citep{ParadaContzen2021a} where the authors show that droop control is inefficient to handle a meshed network topology. Despite the rapid developments of control techniques and deeper system theoretic understanding gained in the past decades, two main problems have not yet been fully understood: safe operation during transients for input currents combined with voltage regulation, and the role of current sharing in a meshed network topology.

Among distributed control techniques, receding horizon controllers offer a methodology that includes constraints in its formulation. Two main stream approaches exist for these kind of controllers \citep{Maestre2014a}. The first kind is based on distributed optimisation techniques such as \cite{Kohler2019} and \cite{Engelmann2020}; these aim to solve the overall optimisation problem in a distributed way. The second type of approaches employ robust control techniques like those of \cite{Trodden2017} and \cite{Riverso2018}; the objective of these approaches lies in handling interconnections as disturbances to be rejected. A common feature of both approaches is that they both require some degree of cooperation between the elements of the network. This is particularly useful when global constraints exists as shown by \cite{Jin2021} and \cite{Wang2017}.The former requires a more intensive level of communication, whereas robust approaches require only a limited exchange of information. This by no means represents a dichotomy, these approaches represent two sides of the same problem and the a natural trade-off between performance and ease of computations. For the \ac{MG} case, receding horizon techniques have been used predominantly in their distributed optimisation form as mentioned in the excellent review of \cite{Hu2021}. Some examples of these approaches are \citep{Zhao2015a, Parisio2014, Hans2018} where the authors exploit the benefits of using receding horizon methods for handling power flows and constraints. Robust control methods, however, have not yet, to the best of the authors' knowledge, established a foothold in a \ac{MG} setting. The source of limitation is the ubiquitous assumption on the size of the interaction strength \citep{Baldivieso2018b}. In a \ac{MG} setting, the interactions represent currents flowing through the network whose magnitude is comparable to that of local states. Therefore, robust control methods for distributed receding horizon control require adjustments to how they handle interactions among individual elements of the network.

In this paper, we aim to rigorously analyse the system theoretic properties of a DC network composed of buck converters arranged in meshed topology where each converter feeds generic nonlinear loads with a mild continuity assumption. To this end, we aim to address two important issues in DC distribution networks, that of operational safety in terms of input currents, together with voltage regulation that offers an equivalent version of load sharing for meshed network configurations while maintaining operation voltages within a compact range of a nominal voltage. Our proposed controller architecture contains a decentralised primary controller ensuring operation safety for inputs currents and a receding horizon distributed controller for voltage regulation. The analysis of the primary controller hinges on extending the approach of \cite{Konstantopoulos2019a} where the authors propose a state-limiting converter based on the concept of \ac{BIC}; we prove asymptotic stability inside a compact set representing a desired safety set using a Lyapunov approach defined on this set. We study the interconnection properties between voltages and currents under the scope of the cascaded systems approach; our results show that the kernel of the network Laplacian matrix defines an attractive set which yields an ISS type result. This analysis is, however, not enough to guarantee the convergence to a particular equilibrium point, as we show with some simple examples. To regulate voltages, we employ a distributed voltage regulation based on concepts of robust distributed model predictive controllers. The latter implies that information sharing occurs only once each sampling period as opposed to distributed approaches. On this vein, we rely and extend on the approach of \cite{BaldiviesoMonasterios2018} which proposes an MPC technique capable of handling exogenous information; in addition to neighbouring voltage information, we assume that a nominal measurement of the load is available at each sampling time. The resulting distributed controller is, to the best of the authors' knowledge, the first attempt to use a non-iterative distributed predictive controller. We analyse the nominal equilibrium behaviour of the proposed control law, and show that, assuming a bounded load deviation, it remains in a neighbourhood of the ``real'' equilibrium, \ie the one considering uncertain loads. Similarly, we show that our distributed control law at steady state lies in the equilibrium manifold of the system.

Our contributions are:
\begin{enumerate}[i)]
\item We propose a Lyapunov function for the \ac{BIC}-based primary controller in closed loop form with each node current as opposed to \cite{Konstantopoulos2019a} where only local stability results are obtained based on linearisation. With this Lyapunov function, we can conclude the asymptotic stability of all the current dynamics. 
\item We show that under the cascaded system interpretation, a buck converter network satisfies a weaker ISS property. The state remains close to the kernel of the Laplacian and not around an equilibrium point.
\item We propose a novel non-iterative distributed receding horizon voltage controller to steer the voltages towards a given equilibrium point which does not rely on the ubiquitous assumption of weak coupling between components, as seen for example in \cite{Riverso2018} and \cite{Trodden2017}. In the proposed controller, each node employs information about the voltage values of its neighbours to compute its own control law and exploit the influence of neighbouring nodes. 
\end{enumerate}
\textbf{\emph{Notation}}:  A \ac{MG} can be seen as a connected undirected graph $\mc{G} = (\mc{V},\mc{E})$ where the set of nodes $\mc{V}$ represents a collection of inverters and local loads; the set of edges $\mc{E}\subseteq\mc{V}\times\mc{V}$ defining the \ac{MG} topology is characterised by the node-edge matrix $\mc{B}\in\Rset^{|\mc{E}|\times|\mc{V}|}$ which for edge $e = (i,j)\in\mc{E}$ involving nodes $i$ and $j$ can be defined as $[\mc{B}]_{ei} = 1$ if node $i$ is the source of $e\in\mc{E}$, and $[\mc{B}]_{ej} = -1$ if node $j$ is its sink, and zero otherwise. An $s-$partite graph is a graph whose vertices can be partitioned into $s$ disjoint sets $\mc{V} = \bigcup_{h=1}^{s}\mc{V}_h$ such that for each $h$ and $i,j\in\mc{V}_h$, $(i,j)\notin\mc{E}$. The $2-$norm is denoted $|x| = \norm{x}_2$; the distance of a point $x\in\Rset^n$ to a set $\mc{A}\subset\Rset^n$ is denoted $|x|_\mc{A} = \inf\{|x-y|\colon y\in\mc{A}\}$. For a vector $x\in\Rset^n$, the operator $[\cdot]$ denotes $[x] = \diag{x_1,\ldots,x_n}\in\Rset^{n\times n}$. A set $\mc{A}\subset\Rset^n$ is a \emph{C-set} if it is convex and compact; \emph{PC-set} is a C-set with the origin in its nonempty interior. A class $\mc{K}-$function $\alpha\colon\Rset^+\to\Rset^+$ is a continuous, strictly increasing with $\alpha(0)=0$; $\alpha(\cdot)$ is $\mc{K}_\infty$ if in addition $\lim_{r\to\infty}\alpha(r) = \infty$. A function $\beta\colon\Rset\times\Rset\to\Rset^+$ is $\mc{K}\mc{L}$ if $\beta(\cdot,t)$ is $\mc{K}$ for all $t\in\Rset^+$; $\beta(r,\cdot)$ is continuous, strictly decreasing, and $\lim_{t\to\infty}\beta(r,t) = 0$ for all $r\in\Rset^+$.
\begin{definition}
A set $\mc{R}\subset\mbb{X}$ is \emph{control invariant} for $\dot{x} = f(x,u)$ and constraint sets $(\mbb{X},\mbb{U})$ if for any $x_0\in\mc{R}$, there exists a control law $\mu\colon\mc{R}\to\mbb{U}$, such that the closed loop system $\dot{x} = f(x,\mu(x))$ satisfies $x(t)\in\mc{R}$ for all $t\geq 0$ with $x(0) = x_0$.
\label{def:ci}	
\end{definition}
\begin{definition}A $\mc{C}^{1}$ function $V\colon\Rset^n\to\Rset$ is a \emph{control Lyapunov function} for a system $\dot{x} = f(x,u)$ with $u\in\mbb{U}$ if it is positive definite $V(x)\geq 0$ for all $x\in\Rset^{n}$, radially unbounded $\lim_{|
  x|\to\infty} V(x) = \infty$, and there exists $\alpha_V\colon\Rset\to\Rset$ nondecreasing and radially unbounded such that
\begin{equation}
  \inf_{u\in\mbb{U}}\pdiff{V}{x}f(x,u) + \alpha_V(|x|) < 0
  \label{eq:clf}
\end{equation}
\label{def:clf}
\end{definition}
\begin{definition}[Continuity for set-valued maps]
  A set $\Phi\colon U\to 2^X$ is
  \begin{enumerate}[i)]
  \item \emph{Upper semi-continuous} (u.s.c) at $x_0\in U$ if for an open neighbourhood $V_U\subset U$ of $x_0$, for all  $x\in\ V_U$,  $\Phi(x)\subset V_X$ for an open neighbourhood $V_X\subset X$.\label{def:usc}
  \item \emph{Lower semi-continuous} (l.s.c) at $x_0\in U$ if for any $y_0\in\Phi(x_0)$ and a neighbourhood $V_X\subset X$, there exists $V_U\subset U$ such that for all $x\in V_U$, $\Phi(x)\cap V_X \neq\emptyset.$\label{def:lsc}
  \item \emph{Continuous} if it is u.s.c. and l.s.c.\label{def:cont}
  \end{enumerate}
\label{def:continuous}  
\end{definition}
%
\section{Preliminaries}
\label{sec:preliminaries}
\subsection{Modelling of DC-DC buck converters}
Consider $M$ DC-DC buck converters connected through $M_{E}$ power lines. The $i^\ts{th}$ converter voltage and current $(v_i,i_i)$ are modelled as follows:
\begin{subequations}
\begin{align}
		L_{i} \frac{d i_{i}}{dt} &=  -r_{i}i_{i}-v_{i}+\bar{v}_{i},\\
		C_{i} \frac{d v_{i}}{dt} &= -f_{L,i}(v_i,P_{L,i}) +i_{i} - \mc{B}_i^\top i_\mc{E},
\end{align}
\label{buck}
\end{subequations}
where $L_{i}$ and $r_{i}$ are the converter inductance and its associated parasitic resistance; $C_{i}$ denotes the converter capacitance; $f_{L,i}(\cdot)$ represents the current drawn by the local load according to the power reference $P_{L,i}$; and $\bar{v}_i = V_i^\ts{IN}\tilde{m}_i$ where $V_i^\ts{IN}$ is the input voltage and $\tilde{m}_i$ is the converter duty ratio. The network affects the $i^\ts{th}$ converter via $\mc{B}_i^\top i_\mc{E}$ with $\mc{B}_i^\top$ corresponding to the $i^\ts{th}$ column of the node-edge matrix $\mc{B}$ and $i_\mc{E}$ the current running through the network lines. The dynamics of current flowing through a power line $e = (i,j)$ connecting the converter $i$ to $j$ is
\begin{equation}
	\begin{split}
		L_{e} \frac{d i_{e}}{dt} &= -r_e i_e + \mathcal{B}_e v,
	\end{split}
\end{equation}
where $L_e$ and $r_e$ are the inductance and resistance of line $e$. Similarly to the previous case, the current flowing through line $e = (i,j)\in\mc{E}$ depends on the voltages of node $i$ and $j$ characterised by $\mc{B}_e$ which is the $e^\ts{th}$ row of $\mc{B}$. 



\subsection{Current control Structure} 
The decentralised primary current controller follows the structure  developed in \cite{Konstantopoulos2019a} which ensures current limitation. The main modification is that the controller proposed in \cite{Konstantopoulos2019a} guarantees $|i_i|\leq I_i^\ts{max}$, where $I_i^\ts{max}$ is associated with the rated capacity, but in the context of buck converters, the input current cannot be negative, \ie only one directional power flows are allowed. To overcome this problem, we introduce a shift of coordinates $ i_i = \tilde{i}_i + i_{s,i}$ where $i_{s,i} = \frac{1}{2}I_i^\ts{max}$, then the proposed nonlinear PI controller that allows tracking of a current reference $i_i^\ts{ref}$ is 
\begin{subequations}
\begin{align}
		\bar{v}_i& = v_i - k_{P,i}\tilde{i}_{i} + r_ii_{s,i} + M_{i} \sin(\sigma_{i}),\label{eq:control_law}\\
		\frac{d \sigma_{i}}{dt} &= \frac{k_{I,i}}{M_i}(i_{i}^\ts{ref}-\tilde{i}_{i} - i_{s,i}) \cos(\sigma_i),\label{eq:int_dyn}
\end{align}
\label{eq:cont}
\end{subequations}
where $k_{P,i}>0$, $k_{I,i}>0$, $M_i = \frac{1}{2}(r_i+k_{P,i})I_{i}^\ts{max}$, are controller gains to ensure the $i^\ts{th}$ inverter current satisfies $|\tilde{i}_i(t)|\leq \frac{1}{2}I_{i}^\ts{max}$. The closed-loop dynamics with the nonlinear control law \eqref{eq:cont} are described by the following equations for all $i\in\mc{V}$ and $e\in\mc{E}$:
\begin{subequations} 
	\begin{align}
		L_{i} \frac{d \tilde{i}_{i}}{dt} &= -(r_{i}+k_{P,i})\tilde{i}_{i} +M_i \sin\sigma_i,\label{eq:current}\\
		M_{i} \frac{d \sigma_{i}}{dt} &= k_{I,i}(i_{i}^\ts{ref}-\tilde{i}_{i}-i_{s,i}) \cos\sigma_i,\label{eq:bic}\\
		C_{i} \frac{d v_{i}}{dt} &= -f_{L,i}(v_i,P_{L,i})+ i_{s,i} + \tilde{i}_i -\mc{B}_i^\top i_\mc{E},\label{eq:voltage}\\
		L_{e} \frac{d i_{e}}{dt} &= -r_{e} i_{e} + \mathcal{B}_{e}v\label{eq:lines}.
	\end{align}
	\label{eq:cl}
\end{subequations}
The closed loop equilibrium points are given by the solution of\footnote{The saturation function is defined as  $\ts{sat}(x,y) = \begin{cases} x & |x|\leq y\\ y~\ts{sgn}(x) & |x|> y\\ \end{cases}$} 
\begin{subequations}
	\begin{align}
		\tilde{i}_i^\ts{eq} &= \ts{sat}(i_i^\ts{ref} - i_{s,i},\frac{1}{2}I_{i}^\ts{max}),\label{eq:eq_current}\\	
		\sigma_i^\ts{eq} &= \ts{sat}\biggr ( \arcsin{\frac{2(i_{i}^\ts{ref} - i_{s,i})}{I_{i}^\ts{max}}},\frac{\pi}{2}\biggr)\label{eq:eq_bic},\\
		\tilde{i}^\ts{eq} + i_{s,i}&= f_{L}(v^\ts{eq},P_{L}) +\mc{B}^\top{r}_{E}^{-1}\mc{B}v^\ts{eq},\label{eq:eq_voltage}\\
		i^\ts{eq}_\mc{E}&={r}_{E}^{-1}\mc{B}v^\ts{eq},\label{eq:line_equi}
	\end{align}
	\label{eq:equilibrium}
\end{subequations}
where $f_L(v,P_L) = (f_{L,1}(v_1,P_{L,1}),\ldots, f_{L,M}(v_M,P_{L,M}))$ and $r_E=diag(r_1,\dots,r_{M_E})$. Because of the voltage-current decoupling, the equilibrium point has two distinct parts: one depending exclusively on local information, \ie current and integrator, and network wide equilibria corresponding to both voltages and line currents.
\begin{remark}
The operator $\pdiff{f_L}{v} +\mc{B}^\top{r}_E^{-1}\mc{B}$ is invertible as long as $\pdiff{f_L}{v}$ is positive definite, \ie for passive loads. The Laplacian matrix $\mc{B}^\top{r}_E^{-1}\mc{B}$ is positive semi-definite with a zero eigenvalue. Adding $\pdiff{f_L}{v}$ leads to an eigenvalue shift which implies, following Weyl's inequalities, that  $\pdiff{f_L}{v} +\mc{B}^\top{r}_E^{-1}\mc{B}$ has eigenvalues different than zero.
\label{rem:Load_resistive}
\end{remark}
We make a further assumption on the load characteristics:
\begin{assumption}
For each $i\in\mc{V}$, the current drawn by each load can be written as $f_{L,i}(v_i,P_{L,i}) = g_{L,i}(v_i)d_{L,i}$ with $g_{L,i}\colon\Rset^n_{L,i}\to\Rset$ and some $d_{L,i}\in\Rset^{n_{L,i}}$ defining the load characteristics. In addition, the function $g_{L,i}(\cdot)$ is $\mc{C}^1$ in an open neighbourhood $A\subset\Rset$ of a voltage $v\neq 0$.
\label{assum:load_structure}	
\end{assumption}
\begin{remark}
	Assumption~\ref{assum:load_structure} is not restrictive; for example a \emph{ZIP} load can be written as:
\begin{equation}
  i_{L,i} = \frac{1}{R_{L,i}}v_i + I_{L,i}  + \frac{P_{L,i}}{v_i}
\label{eq:zip_load}  
\end{equation}
for an impedance $R_{L,i}$, constant current $I_{L,i}$, and constant power $P_{L,i}$. Therefore, $g_{L,i}(v_i) = (v_{i},1,v_i^{-1}) $, $d_i = (R_{L,i},I_{L,i},P_{L,i})$, and total load power is $\frac{v_i^2}{R_{L,i}}+I_{L,i}v_i+P_{L,i}$.
\end{remark}
\subsection{Constraints and control objectives}
The buck converter network is subject to constraints on the inputs, currents, and output voltages. The first assumption on~\eqref{eq:cl} is concerned with the time-scale separation between~\eqref{eq:current}--\eqref{eq:voltage} and~\eqref{eq:lines}.
\begin{assumption}[Time-scale separation]
The network parameters satisfy
\begin{equation}
	\min_{i\in\mc{V}}\biggl \{\frac{L_i}{r_i+k_{P,i}},\frac{4C_i P_{L,i}}{(I_i^\ts{max})^2},\frac{r_i+k_{P,i}}{k_{I,i}}\biggr \} \gg \max_{e\in\mc{E}}\biggl \{\frac{L_e}{r_e}\biggr\}
\end{equation} 
where $P_{L,i}$ is the load power at node $i\in\mc{V}$.
\label{assump:time-scale}
\end{assumption}
Assumption~\ref{assump:time-scale} essentially requires the time constants for each power line to be sufficiently small; the authors in \cite{Venkatasubramanian1995} argue in favour of similar assumptions. A consequence of Assumption~\ref{assump:time-scale} is that~\eqref{eq:cl} can be reduced to only node dynamics, while line currents satisfy the algebraic relation~\eqref{eq:line_equi}. This Assumption takes into account the case of low voltage microgrids, \ie $L_e = 0$. The overall state for each node is therefore $x_i = (v_i,\tilde{i}_i,\sigma_i)$ with inputs $u_i = i_i^\ts{ref} - i_{s,i}$. The system is subject to constraints on the inputs and states such that for all $i\in\mc{V}$:
\begin{equation}
x_i \in \mbb{X}_i,\quad u_i\in\mbb{U}_i
\end{equation}
which satisfy the following assumption:
\begin{assumption}[Network constraints]
 For each $i\in\mc{V}$, 
 \begin{enumerate}[i)]
 \item the input constraint set $\mbb{U}_i\subset\Rset$ is a PC-set, while the state constraint $\mbb{X}_i$ is a C-set.
 \item the load characteristics $d_{L,i} \in\Rset^{n_{L,i}}$ are constrained to a PC-set $\mbb{D}_i\subseteq\Rset^{n_{L,i}}$.
 \end{enumerate}
\label{assump:constraints}
\end{assumption}
\begin{remark}
	The constraints imposed on each node are generally hypercubes, \ie $[0,V_\ts{in}]\times\frac{1}{2}[-I_i^\ts{max},I_i^\ts{max}]\times [-\frac{\pi}{2},\frac{\pi}{2}]$, and input constraints $\mbb{U}_i =\frac{1}{2} [-I_i^\ts{max},I_i^\ts{max}]$. In our setting, however, the feedback transformation~\eqref{eq:control_law} imposes additional constraints
	\[
	v_i - k_{P,i}\tilde{i}_i + r_i i_{s,i} + M\sin\sigma_i \in [0,V_\ts{in}].
	\]
	which constrains the converter modulation index.
\end{remark}
The aim is to solve the following optimal control problem: from a state $x(0)$, determine the control policy, \ie reference currents, that minimises the criteria  
\begin{equation}
	J(x,u,v^*) = \int_0^\infty (\mathds{1}_{|\mc{V}|}v^* - Cx)^\top (\mathds{1}_{|\mc{V}|}v^* - Cx) + \gamma(u)dt.
\label{eq:overall_cost}
\end{equation}
where the node voltage is an output defined by the linear map $v = Cx$. This criteria encourages node voltages to operate near a common point $v^*$, and each source to feed its associated local load while minimising its operating costs $\gamma(\cdot)$.
\section{Primary controller and interconnection analysis}
\label{sec:inv_stability}
In this section, we analyse the properties of the current controller, and its relation to the network voltage dynamics. First in Section~\ref{sec:inv_stability:primary}, we prove the invariance properties of the bounded integral control, then using a Lyapunov function, we can infer stability using the invariance principle. Later in Section~\ref{sec:inv_stability:cascaded}, the dynamics of each node are decomposed into two constituting parts: the \emph{driving system} given by currents and integrator; and the \emph{driven system} composed by node voltages.

\subsection{Current controller properties}
\label{sec:inv_stability:primary}
The next result is an adaptation to our setting of \cite[Proposition 3]{Konstantopoulos2019a}
\begin{lemma}[Bounded integral control] 
For all $i\in\mc{V}$, the set $\mbb{Z}_i = \frac{1}{2}[-I_i^\ts{max},I_i^\ts{max}]\times [-\frac{\pi}{2},\frac{\pi}{2}]$ is \emph{control invariant} for \eqref{eq:current} and \eqref{eq:bic} with $I_i^\ts{max} = \frac{2M_i}{(r_i+k_{P,i})}$.
\label{lem:bic}
\end{lemma}
\begin{proof}
	Consider the function $W_i(\tilde{i}_i,\sigma_i) = \frac{1}{2}\tilde{i}_i^2$, and its time derivative:
	\[
	\dot{W}_i  = -(r_i + k_{P,i})\tilde{i}_i^2 + M_i\tilde{i}\sin\sigma_i \leq -(1-\gamma_i)(r_i + k_{P,i})\tilde{i}_i^2
	\]
	for all $|\tilde{i}_i| \geq \frac{M_i}{\gamma_i k_{P,i}}$ and $\gamma_i\in (0,1)$ which implies $\tilde{i}_i(t)$ is ultimately bounded. Following \cite[Theorem 4.18]{Khalil2001}, the shifted current satisfies $|\tilde{i}_i|\leq \frac{I_i^\ts{max}}{2}$.
	
	For the integrator we use contradiction, suppose $\sigma_i(0) \in [-\frac{\pi}{2},\frac{\pi}{2}]$ and   $\exists t_u >0$ such that $\sigma_i(t_u) \notin [-\frac{\pi}{2},\frac{\pi}{2}]$. By continuity of solutions, there exist a $t_1$ with $\sigma_i(t_1) = \frac{\pi}{2}$, however this implies that $\cos(\sigma_i(t_1)) = 0$ which stops integration. On the other hand, there are three equilibrium points in $\mbb{Z}_i$ following~\eqref{eq:eq_current} and \eqref{eq:eq_bic} for each $u_i\in\mbb{U}_i$, \ie $z_{i,1}^\ts{eq} = (u_i,\arcsin\frac{2u_i}{I_i^\ts{max}})$ and $z_{i,2,3}^\ts{eq} = (\pm \frac{I_i^\ts{max}}{2},\pm \frac{\pi}{2})$. Using \cite[Theorem 1]{Konstantopoulos2019a}, the equilibria at the boundary of $\mbb{Z}_i$ is unstable. For the case, $u_i = \pm I_i^\ts{max}$, the associated equilibrium is $\sigma_i^\ts{eq} = \pm\frac{\pi}{2}$ which is stable. This implies a contradiction, therefore $|\sigma_i(t)| \leq  \frac{\pi}{2}$ for all $t\geq 0$. As a result, the set $\mbb{Z}_i$ is control invariant for shifted currents and integrator for any $u_i \in \mbb{U}_i$.
\end{proof}
\begin{proposition}[Bounded integrator Lyapunov function]
	For each $i\in\mc{V}$ and any $u_i\in \frac{1}{2}(-I_i^\ts{max},I_i^\ts{max})$, the $\mc{C}^{1}$ function $W_i\colon\Rset^2\to\Rset$  defined as
	\begin{equation}
		\begin{split}
			W_i(\tilde{i}_i,\sigma_i) = & \frac{1}{2}L_i(\tilde{i}_i - u_i)^2 +\\
			& \frac{M_i^2}{k_{I,i}}\biggl (1 - \frac{2u_i}{I_i^\ts{max}}\biggr ) \ln\biggl | \frac{M_i\sqrt{1 - \bigl ({\frac{2u_i}{I_i^\ts{max}}}\bigr )^2}}{\cos\sigma_i}\biggr | +\\
	 & \frac{M_i^2}{k_{I,i}}\frac{2u_i}{I_i^\ts{max}}\ln\biggl | \frac{1+\frac{2u_i}{I_i^\ts{max}}}{1+\sin\sigma_i}\biggr |
		\end{split}
		\label{eq:current_Lyapunov}
	\end{equation}
	is a Lyapunov function for the driving subsystem inside $\ts{int}(\mbb{Z}_i)$.
	\label{prop:lyapunov_driving}
\end{proposition}
\begin{proof}
The time derivative of~\eqref{eq:current_Lyapunov} yields:
\[
\begin{split}
	\dot{W}_i = & \pdiff{W}{\tilde{i}_i}f(\tilde{i},\sigma_i) + \pdiff{W}{\sigma_i}g(\tilde{i},\sigma_i,u_i)\\
	= &~~~ (\tilde{i}_i - u_i)( -(r_{i}+k_{P,i})\tilde{i}_{i} +M_i \sin\sigma_i)  \\
	& +M_i\biggl (1 - \frac{2u_i}{I_i^\ts{max}}\biggr )\frac{\sin\sigma_i}{\cos\sigma_i} (u_i-\tilde{i}_{i}) \cos\sigma_i  \\
	& -M_i\frac{2u_i}{I_i^\ts{max}}\frac{\cos\sigma_i}{1+\sin\sigma_i} (u_i-\tilde{i}_{i}) \cos\sigma_i 
\end{split}
\]
which results in
\begin{equation}
  \dot{W}_i  < -\gamma_i(r+k_{P,i})(\tilde{i}_i - u_i)^2
  \label{eq:Lyap_decrease}
\end{equation}

with $\gamma_i\in(0,1)$. Note that, following~\eqref{eq:eq_current}, the derivative $\dot{W}_i$ only vanishes at the equilibrium point. For positive definiteness, we analyse the extremums of the candidate function $W_i(\cdot,\cdot)$; the only viable solutions of $\nabla W_i = 0$ are
\[
\tilde{i}_i^\ts{eq} = u_i;\quad \sigma_i^\ts{eq} = \arcsin\frac{2u_i}{I_i^\ts{max}}
\]
This solution represents the equilibrium point given in~\eqref{eq:eq_current} and \eqref{eq:eq_bic}. There is an extra solution at $\sigma = -\frac{\pi}{2}$, but this lies outside the domain of definition of $W_i$ so it is discarded. The Hessian of $W_i$ satisfies
\[
\nabla^2W_i =\begin{bmatrix}
	L_i& 0\\ 0 & \bigl (1 - \frac{2u_i}{I_i^\ts{max}}\bigr )\sec^2\sigma_i + \frac{2u_i}{I_i^\ts{max}}\frac{1}{1+\sin\sigma_i}
\end{bmatrix} > 0
\]
Since $\sec\sigma_i > 0$ and $1+\sin\sigma_i >0$ for $\sigma_i\in \bigl ( -\frac{\pi}{2},\frac{\pi}{2}\bigr )$, then $W_i(\tilde{i}_i^\ts{eq},\sigma_i^\ts{eq}) \leq W_i(\tilde{i}_i,\sigma_i)$ for all $(\tilde{i}_i,\sigma_i)\in\ts{int}(\mbb{Z}_i)$. As a result, the candidate function is a Lyapunov function.
\end{proof}
An immediate consequence of the above result is
\begin{corollary}[Asymptotic stability driving subsystem]
	 Suppose Assumption~\ref{assump:constraints} holds. For all $i\in\mc{V}$ and any $u_i\in \frac{1}{2}(-I_i^\ts{max},I_i^\ts{max})$, $|(\tilde{i}_i(t),\sigma_i(t))|_{(\tilde{i}_i^\ts{eq},\sigma_i^\ts{eq})}\to 0$ as $t\to\infty$.
	\label{cor:Lyapunov_stability}
\end{corollary}
\begin{proof}
The proof relies on the LaSalle's invariance theorem. To this aim, let 
\[\Omega_{i,(c,u_i)} = \{(\tilde{i}_i,\sigma_i)\colon W_i(\tilde{i}_i,\sigma_i) \leq c\} \]
be a $c-$level set for $W_i$. We claim $\Omega_{i,(c,u_i)}$ is a PI set for all $|u_i|< \frac{1}{2}I_i^\ts{max}$. The normal vector, the gradient $\nabla W_i$, of each level set forms a negative angle with the vector field, $f_i$, defining the dynamics.
\[	\nabla W_i f_i(\tilde{i}_i,\sigma_i,u_i) \leq 0 \]	
Next, by construction, we note that $W_i(\tilde{i}_i,\sigma_i) =0$ only at the equilibrium $(\tilde{i}_i^\ts{eq},\sigma_i^\ts{eq})\in \Omega_{i,(c,u_i)}$. As a result, the asymptotic stability of the equilibrium point follows by applying the invariance principle.
\end{proof}
Both Proposition~\ref{prop:lyapunov_driving} and Corollary~\ref{cor:Lyapunov_stability} establish asymptotic stability inside the constraint set $\mbb{Z}_i$; furthermore, each level set $\Omega_{i,(c,u_i)}\subseteq\mbb{Z}_i$ of $W_i$ is also positively invariant. The case for which $u_i = \pm \frac{I_i^\ts{max}}{2}$ corresponds to a Lyapunov function
\[
W_i(\tilde{i},\sigma_i) = \frac{1}{2} L_i \tilde{i}_i^2 + \frac{M_i^2}{k_{I,i}}\ln \biggl | \frac{2}{1\pm\sin\sigma_i}\biggr |.
\]
This corresponds to a limiting process for $u_i\to\pm I_i^\ts{max}$. The Lyapunov function~\eqref{eq:current_Lyapunov} can be decomposed into two terms corresponding to currents and integrator states, \ie $W_i = U_i(\tilde{i},u_i) + \Phi_i(\sigma_i,u_i)$. The term associated with the energy of the inductor is defined globally, but the one corresponding to the integrator is only defined for the constraint set $\mbb{Z}_i$, see Fig.~\ref{fig:Lyapunov_function}. This is a consequence of the multi-stability properties of \eqref{eq:bic}; those equilibrium points corresponding to $(\pm I_i^\ts{max},\pm \frac{\pi}{2} \pm \pi k )$ for $k \in \{0,1,\ldots\}$ for $u_i\in \frac{1}{2}(-I_i^\ts{max},I_i^\ts{max})$ are not stable which ensures that the basin of attraction remains inside $\mbb{Z}_i$, see Fig.~\ref{fig:Lyapunov_function_W}. The following corollary presents an alternative form for the bounded integrator~\eqref{eq:int_dyn} by means of a diffeomorphism:
\begin{corollary}
The mapping $\Phi\colon\Rset^2\to\Rset^2$ given by $\Phi(\tilde{i},\sigma_i) = (\tilde{i}_i,\sin\sigma_i)$ is a diffeomorphism in $\Rset\times [-\frac{\pi}{2},\frac{\pi}{2}]$.
\label{lem:change_coordinates}
\end{corollary}
Using the change of coordinates defined in Lemma~\ref{lem:change_coordinates}, the current-integrator dynamics have polynomial dynamics for each $i\in\mc{V}$:
\begin{subequations}
\begin{align}
		L_{i} \frac{d \tilde{i}_{i}}{dt} &= -(r_{i}+k_{P,i})\tilde{i}_{i} +M_i \tilde{\sigma}_i,\label{eq:current_m}\\
		M_{i} \frac{d \tilde{\sigma}_{i}}{dt} &= k_{I,i}(u_{i}-\tilde{i}_{i}) (1 - \tilde{\sigma}^2_i).\label{eq:bic_m}
\end{align}
\label{eq:current_integrator_changed}	
\end{subequations}
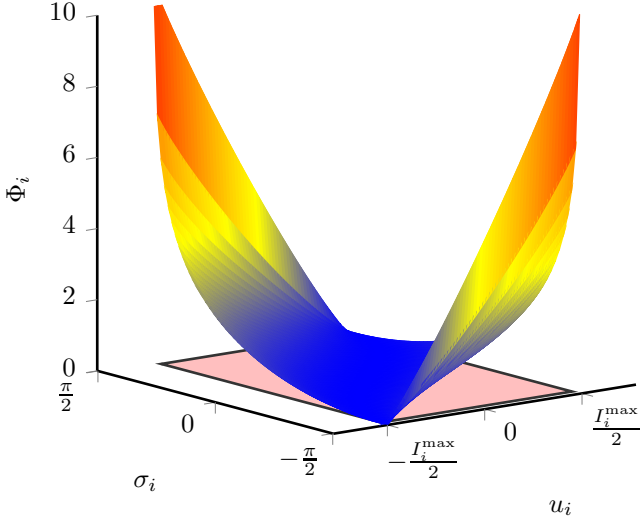
\begin{figure}[t!]
	\centering
	\definecolor{mycolor1}{rgb}{0.00000,0.44700,0.74100}%

\begin{tikzpicture}
\begin{axis}[%
enlargelimits=false,
width=0.8\linewidth,
height=0.7\linewidth,
scale only axis,
ymin=-1.58,
ymax=1.58,
ylabel = $\sigma_i$,
ytick = {-1.570796326794897, 0, 1.570796326794897},
yticklabels={$-\frac{\pi}{2}$, $0$, $\frac{\pi}{2}$},
tick align=outside,
xmin=-1.5697963267949,
xmax=1.5697963267949,
xtick = {-1.01, 0, 1.01},
xticklabels={$-\frac{I_i^\ts{max}}{2}$, $0$, $\frac{I_i^\ts{max}}{2}$},
xlabel = $u_i$,
zmin=0,
zmax=10,
zlabel = $\Phi_i$,
view={-37.9775061027938}{12.6433740229138},
axis x line*=bottom,
axis y line*=left,
axis z line*=left,
legend style={at={(1,1.05)},anchor=north east,fill = none,draw =none},
]

\addplot3[area legend,solid,line width=1.0pt,draw=black!80!white,fill=red!50!white,fill opacity=0.5]
table[row sep=crcr] {%
x	y	z\\
0.99 -1.4707963267949	0\\
-0.99	-1.4707963267949	0\\
-0.99	1.4707963267949	0\\
0.99 1.4707963267949	0\\
}--cycle;\label{fig:set_data}

\addplot3[surf,
  draw opacity = 1,
    fill opacity = 1,
    shader=flat,
    mesh/rows=60] file {figs/Lyap_data.txt};\label{fig:Lyap_data}
\end{axis}
\end{tikzpicture}%
	\caption{Second and third terms of the Lyapunov function~\eqref{eq:current_Lyapunov} $\Phi(\sigma_i,u_i)$, as a function of the system inputs $u_i$, and integrator states $\sigma_i$ \eqref{fig:Lyap_data} over its respective domain of definition \eqref{fig:set_data}, \ie $[-\frac{\pi}{2},\frac{\pi}{2}]\times[-I_i^\ts{max},I_i^\ts{max}]$.}
	\label{fig:Lyapunov_function}
\end{figure}
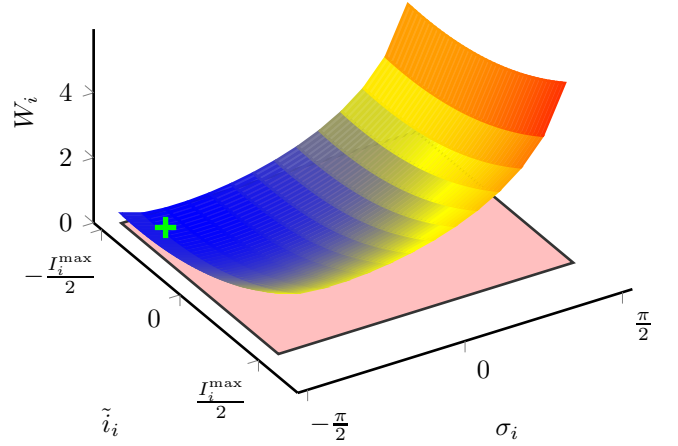
\begin{figure}[t!]
	\centering
	\definecolor{mycolor1}{rgb}{0.00000,0.44700,0.74100}%

\begin{tikzpicture}
\begin{axis}[%
enlargelimits=false,
width=0.8\linewidth,
height=0.7\linewidth,
scale only axis,
xmin=-1.1,
xmax=1.5,
xlabel = $\tilde{i}_i$,
xtick = {-1, 0, 1},
xticklabels={$-\frac{I_i^\ts{max}}{2}$, $0$, $\frac{I_i^\ts{max}}{2}$},
tick align=outside,
ymin=-1.6697963267949,
ymax=1.6697963267949,
ytick = {-1.570796326794897, 0, 1.570796326794897},
yticklabels={$-\frac{\pi}{2}$, $0$, $\frac{\pi}{2}$},
ylabel = $\sigma_i$,
zmin=0,
zmax=5.95987287216692,
zlabel = $W_i$,
view={58.6256198347108}{45.8},
axis x line*=bottom,
axis y line*=left,
axis z line*=left,
legend style={at={(1,1.05)},anchor=north east,fill = none,draw =none},
]

\addplot3[area legend,solid,line width=1.0pt,draw=black!80!white,fill=red!50!white,fill opacity=0.5]
table[row sep=crcr] {%
x	y	z\\
1	-1.4707963267949	0\\
-1	-1.4707963267949	0\\
-1	1.4707963267949	0\\
1	 1.4707963267949	0\\
}--cycle;\label{fig:set_Z}

\addplot3[surf,
  draw opacity = 1,
    fill opacity = 0.9,
    shader=flat,
    mesh/rows=41] file {figs/Lyap_data_shifted.txt};\label{fig:Lyap_shifted}

\addplot3[color=green, mark size=4pt, mark = +,line width = 2pt] coordinates{(-0.8,-1.187121281842605,0)};\label{fig:equilibrium}

\end{axis}
\end{tikzpicture}%
	\caption{The Lyapunov function~\eqref{eq:current_Lyapunov} \eqref{fig:Lyap_shifted} for a constant $u_i = 0.8\frac{I_i^\ts{max}}{2}$ defined over the set $\mbb{Z}_i$ \eqref{fig:set_Z}. Clearly, the equilibrium point~\eqref{fig:equilibrium} is located at the minimum of this function.}
	\label{fig:Lyapunov_function_W}
\end{figure}
\subsection{Cascaded Structure}
\label{sec:inv_stability:cascaded}
We exploit the dynamic structure for each node which can be seen as an interconnection between node voltages and currents with their associated integrators. This implies that the state of each node can be written as: $x_i = (v_i,z_i)$ with $z_i = (\tilde{i}_i,\sigma_i)$ as the driving states. These dynamics are
\begin{subequations}
\begin{align}
  \dot{v}_i = & f_i(v_i,0) + h(v_i,z_i)\\
  \dot{z}_i = & g_i(z_i,u_i)
\end{align}
\label{eq:cascaded_buck}	
\end{subequations}
where $h_i(v_iz_i) = f_i(v_i,z_i) - f_i(v_i,0)$. The stability analysis of cascaded systems has been thoroughly explored, see for example \cite{Sepulchre1997,Isidori1999}, and two main methods exist to infer the asymptotic stability of the cascaded system. These methods rely on satisfying either:
\begin{enumerate}[i)]
	\item Asymptotic stability for the driving system and zero dynamics of the driven system, plus a linear growth restriction on the interaction between these two components. 
	\item \ac{ISS} for the \emph{driven} subsystem and asymptotic stability for the \emph{driving} subsystem.
\end{enumerate}
As a consequence of the results obtained in Section~\ref{sec:inv_stability:primary}, the \emph{driving} subsystem satisfies these stability requirements. We proceed to investigate the stability properties for the \emph{driven}  subsystem. In the following, we consider the overall model, \ie the aggregate model encompassing all nodes; an initial guess of a Lyapunov function, given the linearity of the network, is a quadratic function $S(v) = \frac{1}{2}v^\top Cv$ corresponding to the energy stored in the capacitors $C = \diag{C_i\colon i\in\mc{V}}$. To account for a nonzero equilibrium point, $v^\ts{eq}\in\Rset^{|\mc{V}|}$, and following the approach given in \cite{DePersis2015} and \cite{Alexandridis2019a}, we employ a Bergman Function 
\begin{equation}
  \mc{S}(v) = S(v) - S(v^\ts{eq}) - \pdiff{S}{v}\biggr\vert_{v^\ts{eq}} (v - v^\ts{eq}).
\label{eq:Storage_candidate_voltage}
\end{equation}
Its time derivative yields
\[
\begin{split}
\dot{\mc{S}} & = \pdiff{S}{v} \dot{v} - \pdiff{S}{v}\biggr\vert_{v^\ts{eq}} \dot{v}\\
 & = (v- v^\ts{eq})^\top(-\mc{L}v - g_{L}(v)d_{L} + \tilde{i} + i_s),
\end{split}
\]
where $d_L = (d_{L,1},\ldots,d_{L,|\mc{V}|})\in\prod_{i\in\mc{V}}\mbb{D}_i$ is the collection of load characteristics of the network and $\mc{L}\in\mc{R}^{|\mc{V}|\times|\mc{V}|}$ the connectivity graph Laplacian, \ie $\mc{L} = \mc{B}^\top R_E^{-1}\mc{B}$. We are interested in assessing the stability of the voltage zero dynamics. To this end, we set $\tilde{i}$ to $0$. The resulting derivative yields
\[
\dot{\mc{S}}\bigl\vert_{\tilde{i} = 0} = (v- v^\ts{eq})^\top(-\mc{L}v - g_{L}(v)d_{L} + i_s).
\]
When the load is resistive, \ie $g_L(v)d_L = \diag{R_{L,i}^{-1}}v$, the equilibrium point $v^\ts{eq} = (\mc{L} + \diag{R_{L,i}^{-1}})^{-1}i_s$ is clearly stable following Remark~\ref{rem:Load_resistive}. This result does not hold for the general case; for example, when the network contains constant power loads $g_L(v)d_L = \diag{v_i^{-1}}d_L$, the impedance is locally negative, \ie $\pdiff{g_L}{v}d_L = -\diag{v_i^{-2}}d_L$. As a result, the origin becomes a singularity point; and there are multiple unstable and stable equilibria. However, all is not lost; the voltage dynamics have a Laplacian structure that can be exploited. In the following, we show that a neighbourhood of the synchronisation manifold for the networked voltages is attractive. The network dynamics are dominated by the Laplacian of the connectivity graph. The kernel of which, $\ker{\mc{L}} = \ts{span}~\mathds{1}_{|\mc{V}|}$, determines the synchronisation manifold. In the absence of loads, the kernel is stable \cite{Bullo2020}; in fact in the absence of loads the state converges to a weighted average of the initial state, \ie  \[v\to \frac{\sum_{i\in\mc{V}} C_iv_i(0)}{\sum_{i\in\mc{V}} C_i}\mathds{1}_{|\mc{V}|}\in \ker{\mc{L}}.\]
We are interested in the state evolution $v(t)$ and its properties as $t\to\infty$. The next ancillary result provides estimates for $v(\cdot)$ subject to bounded perturbations.
\begin{lemma}[Network dynamics]
  The set $\mc{R}_\eta = \ker{\mc{L}}\oplus \eta\mbb{B}_{|\mc{V}|}$ is attractive with $\eta >0$ for Laplacian dynamics $C\dot{v} = -\mc{L}v + u$ with $|u(t)| \leq B_u $ for $B_u>0$ and all $t\in\Rset$
\label{lem:network_dynamics}  
\end{lemma}
\begin{proof}
  Given the Laplacian matrix pencil $\lambda C + \mc{L}$ with eigenvalues $\{\lambda_0,\ldots,\lambda_{|\mc{V}|-1}\}$ and the state transformation $y_i = \langle \xi_i,Cv\rangle$ with $\xi_i\in\Rset^{|\mc{V}|}$ the eigenvector associated with $\lambda_i$ for $i\in\{0,\ldots,|\mc{V}|-1\}$. The dynamics of each $y_i$ satisfy 
\[
\dot{y}_i = -\lambda_i y_i + \underbrace{\langle \xi_i,u\rangle}_{\tilde{u}_i}.
\]
The solution for each $i\in\{0,\ldots,|\mc{V}|-1\}$ is
\[
y_i(t) = e^{-\lambda_i t}y_i(0) + \int_0^t e^{-\lambda_i (t-\tau)}\tilde{u}_id\tau.
\]
which yields the explicit expression for the state $v = \sum y_i\xi_i$ such that\[\begin{split}
     v = & \biggl (\frac{\langle\mathds{1}_{|\mc V|}, C{v}(0)\rangle}{\mathds{1}_{|\mc V|}^\top C\mathds{1}_{|\mc V|}} + \int_0^t\frac{\langle\mathds{1}_{|\mc V|}, u(\tau)\rangle}{\mathds{1}_{|\mc V|}^\top C \mathds{1}_{|\mc V|}}d\tau\biggr)\mathds{1}_{|\mc V|} + \\
    & \sum_{i=2}^{n} \int_0^t e^{-\lambda_i (t-\tau)}\langle \xi_i,u(\tau)\rangle d\tau \xi_i.
  \end{split}\]Given that each $z\in\ker\mc{L}$ can be written as $z = \alpha \mathds{1}_n$, the distance from any state to this set is
\[d(v(t),\ker{\mc{L}}) = \min\{|v(t) - \alpha\mathds{1}_n|\colon \alpha\in\Rset\}.\]
The explicit minimum occurs at \[\alpha^*(v) = \bar{v}= \implies d(v,\ker\mc{L}) = |v - \bar{v}\mathds{1}_n|.\]From the state evolution, it is possible to conclude that
\[d(v,\ker\mc{L}) = |(\Xi_{-1} - \tilde{\Xi})s(t)|,\]
where $\Xi_{-1}$ is the matrix of eigenvectors minus the first column, $\tilde{\Xi} = (\bar{\xi}_2,\ldots,\bar{\xi}_{n})\otimes \mathds{1}_n$ a matrix of eigenvalue averages, and the $i^\ts{th}$ component of $s\in\Rset^{n-1}$ is
\[s_i = \int_0^t e^{-\lambda_i (t-\tau)}\langle \xi_i,u(\tau)\rangle d\tau.\]The input boundedness hypothesis implies \[|s_i|\leq \int_0^t|e^{-\lambda_i(t-\tau)}||\xi_i|B_ud\tau.\] The desired bound is therefore

\begin{equation*}
    \begin{split}
  d(v(t),\ker{\mc{L}}) &\leq |(\Xi_{-1} - \tilde{\Xi})||s|\\  
  & \leq |(\Xi_{-1} - \tilde{\Xi})| \sum_{i=2}^n B_u \int_0^t|e^{-\lambda_i(t-\tau)}|d\tau 
\end{split}  
\end{equation*}
given that the exponential is non-negative 
\[d(v(t),\ker{\mc{L}}) \leq B_u|(\Xi_{-1} - \tilde{\Xi})|\sum_{i=2}^{n-1}\frac{1 - e^{-\lambda_i t}}{\lambda_i}.\]
Clearly, in steady state we have
\[d(v_{ss},\ker{\mc{L}}) \leq B_u \sum_{i=2}^{n-1}\frac{1}{\lambda_i}.\] Setting $\eta = B_u\sum_{i=2}^{n-1}\frac{1}{\lambda_i}$ implies the set $\mc{R} = \ker{\mc{L}}\oplus\eta\mbb{B}_{|\mc{V}|}$ is attractive for $\dot{v} = \mc{L}v + u$. 
\end{proof}
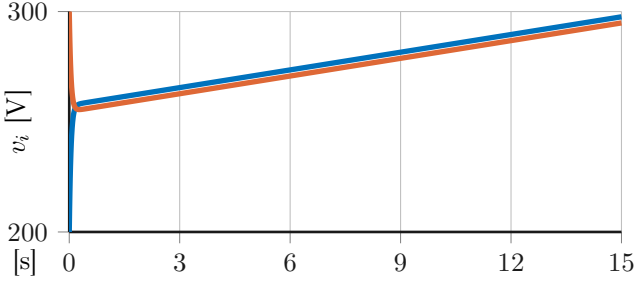
\begin{figure}[t!]
  \centering
%
\definecolor{mycolor1}{rgb}{0.00000,0.44706,0.74118}%
\definecolor{mycolor2}{rgb}{0.87059,0.40784,0.21176}%
\definecolor{mycolor3}{rgb}{0.92941,0.69412,0.12549}%
\definecolor{mycolor4}{rgb}{0.46667,0.67451,0.18824}%
\definecolor{mycolor5}{rgb}{0.30196,0.74510,0.93333}%
\definecolor{mycolor6}{rgb}{0.49412,0.18431,0.55686}%
\begin{tikzpicture}

\begin{axis}[%
width=0.82\linewidth,
height=0.33\linewidth,
at={(0\linewidth,0\linewidth)},
scale only axis,
every outer x axis line/.append style={white!10!black},
every x tick label/.append style={font=\color{white!10!black}},
xmin=0,
xmax=15,
xtick={  0, 3, 6, 9, 12,15},
xticklabels={  $0$, $3$, $6$, $9$, $12$,$15$},
tick align=outside,
xlabel={[\si{\second}]},
xmajorgrids,
every outer y axis line/.append style={white!10!black},
every y tick label/.append style={font=\color{white!10!black}},
ymin=200,
ymax=300,
ytick={200,300},
ylabel={$v_i$ [\si{\volt}]},
ymajorgrids,
y label style={at={(axis description cs:-0.05,0.5)}},
x label style={at={(axis description cs:-0.08,-0.03)}},
axis x line*=bottom,
axis y line*=left
]
\addplot [color=mycolor1,solid,line width = 2.0pt]
table[row sep=crcr]{%
0	200\\
0.00849454526199631	209.716401029707\\
0.0169890905239926	217.797492975808\\
0.0254836357859889	224.519150566123\\
0.0339781810479852	230.110708784879\\
0.0443236827755937	235.671333791996\\
0.0546691845032022	240.116378439806\\
0.0650146862308107	243.670772348052\\
0.078186638369534	247.190853123238\\
0.0913585905082573	249.841310719044\\
0.107370819866965	252.204270805759\\
0.123383049225674	253.882793121841\\
0.143607524583573	255.33565380381\\
0.163831999941473	256.28800761636\\
0.190874202583546	257.084376144723\\
0.223413489020356	257.628377200039\\
0.264277168526499	257.997134093717\\
0.318490921723496	258.265350872312\\
0.396965137237386	258.517705930333\\
0.529606263191579	258.879526728935\\
0.847901879520131	259.73121316862\\
2.16392870097362	263.254483634168\\
3.4799555224271	266.777661117825\\
4.79598234388059	270.300850076342\\
6.11200916533408	273.824037618751\\
7.42803598678756	277.347225335921\\
8.74406280824105	280.870413031524\\
10.0600896296945	284.393600729788\\
11.376116451148	287.916788427725\\
12.6921432726015	291.439976125701\\
14.008170094055	294.963163823673\\
15	297.618431076697\\
};\label{fig:vr:v_1}

\addplot [color=mycolor2,solid,line width = 2.0pt]
table[row sep=crcr]{%
0	300\\
0.00849454526199631	292.377613495876\\
0.0169890905239926	286.044951065234\\
0.0254836357859889	280.784437766859\\
0.0339781810479852	276.415203297384\\
0.0443236827755937	272.079228371747\\
0.0546691845032022	268.623081522482\\
0.0650146862308107	265.869366838118\\
0.078186638369534	263.156248529832\\
0.0913585905082573	261.128978518082\\
0.107370819866965	259.342050766522\\
0.123383049225674	258.094920584996\\
0.143607524583573	257.045934090634\\
0.163831999941473	256.391683835986\\
0.190874202583546	255.893101833346\\
0.223413489020356	255.619878172304\\
0.264277168526499	255.524726429891\\
0.318490921723496	255.572795153558\\
0.396965137237386	255.749545478149\\
0.529606263191579	256.09934287138\\
0.847901879520131	256.951808104163\\
2.16392870097362	260.474930525585\\
3.4799555224271	263.998126279312\\
4.79598234388059	267.521312983127\\
6.11200916533408	271.044500803787\\
7.42803598678756	274.567688486619\\
8.74406280824105	278.090876186459\\
10.0600896296945	281.614063884201\\
11.376116451148	285.137251582202\\
12.6921432726015	288.66043928017\\
14.008170094055	292.183626978143\\
15	294.838894231167\\
};\label{fig:vr:v_2}

\end{axis}
\end{tikzpicture}%
  \caption{Response of a two element network, $C\dot{v} = -\mc{L}v -d + u$, to a constant current load using a mismatched current input. The voltage evolution of $v_1$ is \eqref{fig:vr:v_1} and \eqref{fig:vr:v_2} corresponds to $v_2$ starting from different initial conditions. The loads drain $d = (4,10)\si{\ampere}$ respectively; the input currents are $u = (10,5)\si{\ampere}$. The mismatch between injections and loads results in a voltage response that is akin to a ramp; this response is clearly not stable in the traditional sense, however the voltage trajectories are close to the kernel of $\mc{L}$.}
  \label{fig:ramp_voltage}
\end{figure}
The proof of the above lemma exploits the linearity of the weighted Laplacian dynamics to conclude the state remains in the vicinity of not a point, but the kernel of the Laplacian. This property resembles an ISS behaviour with respect to the kernel, this however might not imply the solutions are stable in the ordinary sense. For example, if the perturbation $u(t) = u$ is constant in Lemma~\ref{lem:network_dynamics}, then the solutions would contain a ramp which would lie close to the kernel but whose magnitude keeps increasing as seen in Figure~\ref{fig:ramp_voltage}. In the next section, we aim to carefully select the input currents such that the voltages attain a finite unique equilibria. 
\section{Distributed Voltage regulation}
\label{sec:dist_voltagecontrol}
In the previous section, we have studied the stability properties of a DC buck converter network and concluded an ISS-type stability condition with respect to the Laplacian Kernel. In this section, we exploit these stability properties to obtain closed-loop stability and optimal performance with respect to an optimisation cost. First, we define the distributed optimisation problem, then we proceed to analyse both recursive feasibility and stability of the overall network.
\subsection{Optimal control problem}
\label{sec:optim-contr-probl}
The overall control objective, following~\eqref{eq:overall_cost}, aims to minimise an infinite horizon cost. To obtain a tractable solution to this problem, we optimise this cost in a distributed way; to this aim, the separability of the cost, \ie  $J(x,u) = \sum_{i\in\mc{V}}J_i(x_i,u_i)$, is an implicit assumption. From the dynamics perspective, each node is subject to parametric uncertainty in terms of load demand \ie $d_{L,i}\in\mbb{D}_i$, and an ``artificial'' additive uncertainty arising from the interconnection with neighbouring nodes\footnote{For a graph $\mc{G} = (\mc{V},\mc{E})$, the set of neighbours of node $i$ is $\mc{N}_i=\{j\in\mc{V}\colon (i,j)\in\mc{E}\}$}, \ie $w_i = \sum_{j\in\mc{N}_j}\mc{L}_{ij}v_j$.

To handle parametric uncertainty, we note that Assumption~\ref{assump:constraints}$-ii)$ implies that $d_{L,i} = \sum_{k = 1}^{D_i}\lambda_{i,k}\hat{d}_{k,i}$ where $d_{k,i}$ are the vertices of $\mbb{D}_i$, $\lambda_{k,i} \in [0,1]$, and $\sum_{k=1}^{D_i}\lambda_{k,i} = 1$. Furthermore, it is possible to rewrite the effects of this disturbance with respect to a nominal point $\bar{d}_{L,i} = \frac{1}{D_i}\sum_{k=1}^{D_i}\hat{d}_{k,i}$ such that the voltage dynamics are
\begin{equation}
  \begin{split}
    C_i\dot{v}_i = & -\mc{L}_{ii}v_i - g_{L,i}(v_i)\bar{d}_{L,i} + \tilde{i}_i + i_{s,i}\\
    & + \underbrace{g_{L,i}(v)(\bar{d}_{L,i} - d_{L,i})}_{w_{L,i}(v_i)} + \underbrace{\sum_{j\in\mc{N}_i}\mc{L}_{ij}v_j}_{w_i}.
  \end{split}
  \label{eq:vol_dyn_uncertain}
\end{equation}
Robust methods can be used to handle the parametric uncertainty arising from the loads. These robust  methods, however, are not suitable to handle the additive term arising from the dynamic coupling. The reason behind this is that a robust based distributed MPC requires a weak coupling assumption \cite{Baldivieso2018b}. In systems that have Laplacian dynamics, this is not the case as evidenced by the next example:
\begin{example}
  Consider a two-element network $\mc{G} = \{1,2\}$ such that its voltage constraint set are $\mbb{V}_1 = \mbb{V}_2 = [0,\bar{V}]$. The line interconnecting node $1$ and $2$ has admittance $Y_{12} = 10[\si{\siemens}]$. The effect of node $1$ on node $2$ can be characterised by the set $\mbb{W}_1 = C^{-1}_1Y_{12}\mbb{V}_2$, and vice-versa $\mbb{W}_2 = C^{-1}_2Y_{12}\mbb{V}_1$. The weak coupling assumption for tube MPC methods demands $\mbb{W}_i\subset\mbb{V}_i$ for $i = 1,2$ which is not the case for this simple example. This implies that there are no robust invariant sets capable to account for the interconnection disturbance.
\label{ex:weak_coupl}
\end{example}
The above example illustrates some of the limitations of robust tube approaches for electrical networks. However, by exploiting the Laplacian structure of the network dynamics it is possible to use Lemma~\ref{lem:network_dynamics} to bound the effect of the neighbours of node $i$ have on its dynamics. Given that the voltage of each node operates near the kernel of the Laplacian, we propose a distributed controller that shares current voltage measurements to be used in a finite optimal control problem. The distributed optimisation problem $\mbb{P}_i(x_i,\bar{d}_{L,i},w_i)$ for each node $i\in\mc{V}$ consists in minimising
\begin{equation}
J_i(x_i,u_i)	= \psi_i(x_i(T)) + \int_0^T \underbrace{(v^* - C_ix_i)^2 + \gamma_i( u_i )}_{\ell_i(x_i,u_i)}dt,
\end{equation}
where $x_i = (v_i,\tilde{i}_i,\sigma_i)$, $u_i = i_i^\ts{ref}$, subject to
\begin{subequations}
\begin{align}
  \dot{{x}}_i = F_i({x}_i&,{u}_i,\bar{d}_{L,i})  + E_iw_i,\label{eq:local_constraints:dyn}\\
x_i \in{\mbb{X}}_i, &\quad u_i\in{\mbb{U}}_i,\label{eq:local_constraints:cons}\\
{x}_i(0)  = x_{i}, &\quad {x}_i(T)  \in\mbb{X}_{f,i}.\label{eq:local_constraints:boundary}
\end{align}
\label{eq:local_constraints}
\end{subequations}
Where $E_i\in\Rset^{3\times1}$ determines how the coupling $w_i$ affects the local dynamics. The constraint set is a set-valued map $\mc{U}_i^N(x_i,\bar{d}_{L,i},w_i)\subset\mbb{U}_i$. In this problem, $\psi_i\colon\mbb{X}_{f,i}\to\Rset^+$ is the terminal penalty, and $\mbb{X}_{f,i}\subseteq{\mbb{X}}_i$ the terminal cost which both satisfy the following assumption. 
\begin{assumption}
  For each $i\in\mc{V}$, the terminal conditions for $\mbb{P} _i(x_i,\bar{d}_{L,i},w_i)$ satisfy
  \begin{enumerate}[i)]
  \item (Terminal cost) For all $x\in\mbb{X}_{f,i}$, the terminal cost satisfies
    \begin{equation}
      \min_{\substack{u\in\mbb{U}\\ x_i(t)\in\mbb{X}_{f,i}}}\biggr\{\psi_i(x_i(t)) + \int_0^t\ell_i(x_i,u_i)dt \biggl \} \leq \psi_i(x_i)
      \label{eq:terminal_cost_assump}
    \end{equation}
  \item (Terminal constraints) The set $\mbb{X}_{f,i}\subset\mbb{X}$ is a \emph{control invariant set} for $\dot{x} = F_i(x_i,u_i,d_{L,i})$ with input constraints $\mbb{U}_i\subset\Rset$.
  \end{enumerate}
  \label{assump:terminal_ingredients}
\end{assumption}
Furthermore, the stage cost $\ell_i(\cdot,\cdot)$ satisfy
\begin{assumption}[Positive definite stage cost]
$\ell_{i}\colon\mbb{X}_i\times\mbb{U}_i\to\Rset$ is for each $i \in \mc{V}$ a continuous positive definite functions.
\label{assum:pd}
\end{assumption}
In the context of the \ac{OCP} $\mbb{P}_i(x_i,\bar{d}_{L,i},w_i)$, each local controller does not have access to future values of the interactions $w_i = \sum_{j\in\mc{N}_i}\mc{L}_{ij}v_j$, the only information needed at a time $t>0$ is the measured interaction which is then kept constant along the horizon $[t,t+T]$. The solution at time $t\geq0$ is an optimal control input $u_i\colon[t,t+T]\to\bar{\mbb{U}}_i$ which can be considered as piecewise constant, \ie $u_i(t) = \bar{u}$ for $t \in [0,\delta)$ for $\delta >0$, for tractability purposes. The value $\delta >0 $ acts as a sampling time such that after $u_i(\cdot)$ is applied to the system, at time $t+\delta$ each node measures its state and shares this measurement with its neighbours so that each constructs $w_i(t+\delta)$. The \ac{OCP} is solved with the updated information $x_i(t+\delta)$, $d_{L,i}$, and $w_i(t+\delta)$ to obtain a new piece-wise constant function $u_i\colon[t+\delta,t+\delta+T]\to{\mbb{U}}_i$, and the process is repeated \emph{ad infinitum}.
\subsection{Properties of the \ac{OCP}}
\label{sec:properties-ocp}
In section, we analyse the properties of the \ac{OCP} in terms of robust stability and feasibility. The analysis focuses first on local properties for recursive feasibility and robustness to load changes. Then, we focus on stability properties for the overall network. 
\subsubsection{Recursive feasibility}
\label{sec:recurs-feas}
We invoke the following assumption to make precise the concepts used in the \ac{OCP} formulation.
\begin{assumption}[Information available to the controller] For each $i\in\mc{V}$, 
  \begin{enumerate}
  \item The state $x_i(\delta k)$, interconnection information $w_i(\delta k)$, and nominal load $\bar{d}_{L,i}\in\mbb{D}_i$ are known exactly at time $t = \delta k\geq 0$ with $k\in\Nset$.
  \item There exists $\gamma_{d,i}>0$ such that the ``true'' load $d_{L,i}\in\mbb{D}_i$ satisfies $|\bar{d}_{L,i}-d_{L,i}|<\gamma_{d,i}$.
  \end{enumerate}
\label{assum:preview_information}
\end{assumption}
At a time $\delta k$, the \ac{OCP} is solved by discretising the continuous problem $\mbb{P}_i(x_i,\bar{d}_{L,i},w_i)$ such that the solution is a sequence of $N$ optimal current references $\mb{u}_i^0(x_i,\bar{d}_{L,i},w_i) = \{u_i^0(0),\ldots,u_i^0(N-1)\}$ which corresponds to a piece-wise constant function. The effect of the interconnection $w_i$ is kept constant across the horizon such that $\mb{w}_i = w_i\mathds{1}_{N+1}$; at the next sampling time, this sequence, following Assumption~\ref{assum:preview_information}, is allowed to change. To account for this, the analysis will focus on two scenarios: prove recursive feasibility for an unchanging interconnection, \ie $w_i(k+1) = w_i(k)$. Then proving that recursive feasibility holds for the general case by leveraging on the structure of the \ac{OCP} and the unchanging case.

The first hurdle to overcome is to find suitable terminal ingredients satisfying Assumption~\ref{assump:terminal_ingredients}. The equilibrium point for each node is characterised by the solution of $F_i(x_i,\bar{d}_{L,i},u_i) + E_iw_i = 0$. We note that this equation may admit multiple solutions and is dependent on both the nominal load and the influence from the rest of the network. By construction, the function $F_i\colon\Rset^3\times\mbb{D}_i\times\Rset\to\Rset^3$ is continuously differentiable which implies that there exists a local solution to this problem $(x_i^\ts{eq},u_i^\ts{eq}) = \xi_i^\ts{ss}(\bar{d}_{L,i},w_i)$. The modified cost can be obtained by translating the stage cost to these equilibrium points, \ie $\tilde{\ell}_i(x_i,u_i) = \ell_i(x_i - x_i^\ts{eq},u_i-u_i^\ts{eq})$. The terminal set is computed using the approach used in \cite{Baldivieso-Monasterios2020}; for a given polytopic set $\mbb{X}_{f,i}(x_i^\ts{eq})=\{x\colon H_{f,i}x \leq h_{f_i}\}\subseteq\mbb{X}_i$ with $H_{f,i}\in\Rset^{n_{f,i}\times3}$ and $h_{f,i}\in\Rset^{n_{f,i}}$, it is possible to find a control action $u_i = u_i^\ts{eq} + \kappa_{f,i}(x_i-x_i^\ts{eq})$ such that $x(t)\in\mbb{X}_{f,i}(x_i^\ts{eq})$ for all $t\geq 0$. This control action is built based on the one-step reachability properties of the discretised system and is the solution of
\begin{equation}
    \min \{|u_i-u_i^\ts{eq}|^2\colon H_{i,f}x_i^+\leq \lambda h_{i,f},~u_i-u_i^\ts{eq}\in\mbb{U}_i\},
  \label{eq:terminal_law_ocp}
\end{equation}
where $x_i^+$ is the value of the state at $x_i(\delta (k+1))$, $\lambda \in [0,1]$ is a design parameter to adjust the ``aggressiveness'' of the controller,   and the set of minimisers is $\mc{U}_{f,i}(x_i,\bar{d}_{L,i},w_i)\subset\mbb{U}_i$. The following proposition summarises the properties of the terminal ingredients for the \ac{OCP}.
\begin{proposition}[Shifted terminal ingredients]
Suppose Assumptions~\ref{assum:load_structure} and \ref{assum:pd} hold. $i)$ For a fixed $\bar{d}_{L,i}\in\mbb{D}_i$, if there exists a control action $u_i\in\mbb{U}_i$ such that \[|u_i|\geq \max_{y\in\mbb{X}_{f,i}(x_i^\ts{eq})}\min_{z\in\lambda\mbb{X}_{f,i}(x_i^\ts{eq})}|F_i(y,\bar{d}_{L,i},u) - F_i(0,0,u)-z|,\] then $\mc{U}_i(x_i,bar{d}_{L,i},w_i)\neq \emptyset$ for all $x_i\in\mbb{X}_{f,i}(x_i^\ts{eq})$. Furthermore, $ii)$ the function \[\Psi_{f,i}(x_i) = \inf\{r\in\Rset^+\colon x_i\in r\mbb{X}_{f,i}\}\]is a control Lyapunov function. Lastly, $iii)$ the set $\mbb{X}_{f,i}$ is control invariant for $\dot{x}_i = F_i(x_i,u_i,\bar{d}_i)+E_iw_i$. 
  \label{prop:terminal_ingredients}
\end{proposition}
\begin{proof}
  The proof for $i)$ is based in a modification to our setting of \cite[Theorem 3]{Baldivieso-Monasterios2020}. The second assertion of this proposition follows by construction, when $x_i = x_i^\ts{eq}$, then $\Psi_{f,i}(x_i^\ts{eq}) = 0$. Furthermore, the difference $\Psi_{f,i}(x_i^+) - \Psi_{f,i}(x_i) < 0$ if $u_i\in\mc{U}_i(x_i,\bar{d}_{f,i},w_i)$. In particular, it is possible to construct the function $\alpha_{V,i}(x_i) = \int_t^{t+\delta}\tilde{\ell}(x_i,\kappa_{f,i}(x_i,\bar{d}_i,w_i))dt$ with $\kappa_i(\cdot,\cdot,\cdot)$ a selection of the set-valued map $\mc{U}_i(\cdot)$. Therefore, $\Psi_{f,i}(\cdot)$ satisfies an integral version of Definition~\ref{def:clf}. The control invariance stated in $iii)$ follows from \cite[Corollary 1]{Baldivieso-Monasterios2020}.
\end{proof}
The above proposition shows that our modified terminal conditions satisfy Assumption~\ref{assump:terminal_ingredients}. With the terminal set available, the feasible region for a given horizon $T>\delta>0$ and given interconnection value $w_i$ is defined recursively  as:
\begin{equation}
  \begin{split}
    \mc{X}_i^{\delta (k+1)}(\bar{d}_{L,i},w_i) = & \{x_i\in\mbb{X}_i\colon\exists u_i\in\mbb{U}_i, x_i(\delta)\in\mc{X}_i^{\delta k}\}\\
    \mc{X}_i^{0}(\bar{d}_{L,i},w_i) = & \mbb{X}_{f,i}(x_i^\ts{eq}).
  \end{split}
  \label{eq:feas_region}
\end{equation}
The existence of a feasible set $\mc{X}_i^{N\delta}$ is linked with the existence of a sequence of control actions $\mb{u}_i\in\mbb{U}_i^N$ which is related to the \ac{OCP} solution. The next result states feasibility under unchanging interconnection information.
\begin{proposition}[Recursive feasibility under unchanging $w_i$.]Suppose Assumptions~\ref{assum:load_structure}, and  \ref{assump:constraints}--\ref{assum:preview_information} hold. For each $i\in\mc{V}$, if $\mb{w}_i^+ = \mb{w}_i(t+\delta) = \mb{w}_i(t)$, then $i)$ $x_i\in\mc{X}_i^T(\bar{d}_i,\mb{w}_i)$ implies that $x(t+\delta)\in\mc{X}_i^T(\bar{d}_{L,i},\mb{w}_i^+)$. $ii)$ the set $\mc{X}_i^T(\bar{d}_{L,i},\mb{w}_i)$ is control invariant for $\dot{x}_i = F_i(x_i,u_i,\bar{d}_{L,i}) + E_iw_i$ and $\mbb{U}_i$. 
\label{lem:rec_feasibility_unchaning_disturbance}
\end{proposition}
The proof of this results follows the line of argument of \cite[Proposition 2]{BaldiviesoMonasterios2018} albeit modified to account for the nonlinear nature of the system. This result provides regularity for the value function of the \ac{OCP}, \[V_{N,i}^0(x_i,\bar{d}_{L,i},w_i) = \min\{J_i(x_i,u_i)\colon u_i\in\mc{U}_{i}^N(x_i,\bar{d}_{L,i},w_i)\},\]defined over the set $\mc{Z}_{i}^N(\bar{d}_{L,i}) = \{(x_i,w_i)\colon w_i\in\mbb{W}_i,~x_i\in\mc{X}_i^N(\bar{d}_{L,i},w_i)\}$ where $\mbb{W}_i = \bigoplus_{j\in\mc{N}_i}\mc{L}_{i,j}E_i\mbb{X}_i$. Standard MPC results, see for example~\cite[Chapter 5]{Grune2016a}, \ie feasibility implies stability, lead to the following Corollary
\begin{corollary}[Local Lyapunov function]Suppose Assumptions~\ref{assum:load_structure},~\ref{assump:constraints}--\ref{assum:pd} hold. For each node $i\in\mc{V}$ and a fixed $w_i\in\mbb{W}_i$ and $\bar{d}_{L,i}\in\mbb{D}_i$, if $x_i^\ts{eq}\in\mbb{X}_i$ is an equilibrium of node $i$, then  there exist $\mc{K}$ functions $\alpha_{ih}$ with $h\in\{1,2,3\}$ such that the value function $V_{N,i}^0(\cdot,\cdot,\cdot)$ satisfies
  \begin{equation}
    \begin{split}
      \alpha_{i1}(|x_i-x_i^\ts{eq}|)\leq V_{N,i}^0(x_i,\bar{d}_{L,i},w_i)& \leq \alpha_{i2}(|x_i-x_i^\ts{eq}|)\\
     V_{N,i}^0(x_i^+,\bar{d}_{L,i},w_i) - V_{N,i}^0(x_i,\bar{d}_{L,i},w_i) & \leq -\alpha_{i3}(|x_i-x_i^\ts{eq}|)
    \end{split}
    \label{eq:value_function_lyapunov}
  \end{equation}
\label{cor:value_function_Lyapunov}  
\end{corollary}
The above result allows us to conclude that for a fixed $w_i\in\mbb{W}_i$ and $\bar{d}_{L,i}\in\mc{D}_i$, the system \[\dot{x}_i = F_i(x_i,\kappa_{N,i}(x_i,\bar{d}_{L,i},w_i),\bar{d}_{L,i})+E_iw_i\] is asymptotically stable with respect to $\{x_i^\ts{eq}(\bar{d}_{L,i},w_i)\}$. The next regularity results is useful for proving recursive feasibility for our MPC scheme.
\begin{lemma}[$\mc{K}$-continuity of the value function] Suppose Assumptions~\ref{assum:load_structure}, \ref{assump:constraints}, and \ref{assum:pd} hold, in addition $\bar{d}_{L,i}\in\mbb{D}_i$ is fixed. The value function $V_{N,i}^0(\cdot)$ for each $i\in\mc{V}$ satisfies $|V_{N,i}^0(z) - V_{N,i}^0(\hat{z})| \leq \sigma_V(|z-\hat{z}|)$ over $\mc{Z}^N_i$ and $\sigma_V$ is a $\mc{K}-$function.
\label{lem:value_function_kappa}
\end{lemma}
\begin{proof}
  The claim follows from \cite[Lemma 1]{Limon2009} if the value function is uniformly continuous. The value function $V_{N,i}^0(\cdot,\cdot,\cdot)$ is defined over the compact set $\mc{Z}_i^N(\bar{d}_{L,i})$ which by the Heine-Cantor theorem allows us to conclude uniform continuity from continuity. Our goal is to match the hypothesis of~\cite[Proposition 4.4]{Bonnans2000a} to conclude the continuity of the value function. The first one is continuity of the objective $J_i(\cdot,\cdot)$ which follows from Assumption~\ref{assum:pd}. The second hypothesis asks for the set-valued function $\mc{U}_i^N(\cdot)$ to have a closed graph. This set can be written as $\mc{U}_i^N(x_i,\bar{d}_{L,i},w_i) = \{\mb{u}_i\colon G_i(\mb{u}_i,x_i,w_i,\bar{d}_{L,i})\in\mbb{K}\}$ for a fixed compact set $\mbb{K}$ and a continuous function~$G(\cdot,\cdot,\cdot)$. The closedness of the graph is a consequence of the continuity of dynamics and constraints. The third hypothesis requires that there exists $\alpha\in\Rset$ and a compact set $\mc{C}\subset\mbb{U}^N_i$ such that for every $(x_i,w_i)$ in a neighbourhood of $(\tilde{x}_i,\tilde{w}_i)$, the level set\[\ts{lev}_\alpha J_i(x_i,\cdot) = \{u_i\in\mc{U}_i(x_i,\bar{d}_{L,i},w_i)\colon J_i(x_i,u_i)\leq\alpha\}\]is not empty and contained in $\mc{C}$. This assertion follows from the continuity properties of the cost function; for a fixed $(x_i,w_i)$ the level set is compact for any $\alpha\in\Rset$ and the continuity of the cost guarantees that this property is maintained for any neighbourhood of $(x_i,w_i)$. And lastly, for any neighbourhood $\mc{T}$ of the minimisers $\mc{S}_i(x_i,\bar{d}_{L,i},w_i)$, there exists a neighbourhood $Z$ of $(x_i,w_i)$ such that $\mc{T}\cap\mc{U}_i^N(x_i,\bar{d}_{L,i},w_i) \neq\emptyset$. To prove this statement, we consider two cases an optimal point $u_i^0\in\mc{U}_i^N(x_i,\bar{d}_{L,i},w_i)$ lies either in the interior of $\mc{U}_i^N(x_i,\bar{d}_{L,i},w_i)$ or at the boundary $\mc{U}_i^N(x_i,\bar{d}_{L,i},w_i)$. The claim is trivial for the first case owing to the compacity of the constraints; for the latter, each neighbourhood $V_U$ of the optimal point $u_i^0$ contains points that lie in the interior. For any optimal point at the boundary $u_i^0$ and $(x_i,w_i)$, there exists a neighbourhood $V_U$ such that $V_U\cap\mc{U}_i^N(x_i,\bar{d}_{L,i},w_i)\neq\emptyset$. Given the graph of $\mc{U}_i^N(\cdot)$ is a closed set, there exist a sequence $\{(x_i^{k},w_i^{k},u_i^{k})\}$ converging to $(x_i,w_i,u_i^0)$. By definition of convergence, all up to a finite amount elements of $\{(x_i^{k},w_i^{k},u_i^{k})\}$ lie in the open neighbourhood $Z\times V_U$. Since $\mb{U}_i^N(\cdot)$ is closed, then the closure of $Z\times V_U$ is compact. On the other hand, there exists $Z^k\times V_U\subset \ts{cl}Z\times V_U$ for each element of the sequence that forms a covering of $\ts{cl}Z\times V_U$. The compacity of the later set implies the existence of a finite covering such that for any element of $(\tilde{x}_i,\tilde{w}_i)\in Z = \bigcap_{h=1}^H Z^{k_h}$, $V_U\cap\mc{U}_i^N(\tilde{x}_i,\bar{d}_{L,i},\tilde{w}_i)\neq\emptyset$ 
\end{proof}
The next step in order to prove recursive feasibility, we need to investigate the effect of a changing disturbance, \ie $w_i(k+1) \neq w_i(k)$. To this aim, we invoke the following assumption:
\begin{assumption}[Bounded interconnection]
  For each $i\in\mc{V}$, the interconnection effect at a time $t + \delta$ satisfy $w_i^+ = w_i(t+\delta) = w_i(t) + \Delta w_i$ where $\Delta w_i\in\Delta\mbb{W}_i$. The set $\Delta\mbb{W}_i$ is chosen such that $\lambda_i = \max\{|w-\tilde{w}|\colon w,\tilde{w}\in\mbb{W}_i,~w-\tilde{w}\in\Delta\mbb{W}_i\}$ satisfies
  \[\lambda_i \leq \sigma_V^{-1}\circ\alpha_{i3}\circ\alpha_{i2}^{-1}(\beta_i)\]where $\beta_i> 0$.
\label{assum:bounded_disturbance}  
\end{assumption}

The following result asserts recursive feasibility
\begin{theorem}[Recursive Feasibility]
  Suppose Assumptions~\ref{assum:load_structure}, and  \ref{assump:constraints}--\ref{assum:bounded_disturbance} hold. If at time $t>0$ and for a fixed $\bar{d}_{L,i}\in\mbb{D}_i$ the state satisfies $x_i\in\mc{X}_i^N(w_i,\bar{d}_{L,i})$, then at time $t+\delta$, the state satisfies $x_i(t+\delta)\in\mc{X}_i^N(w_i^+,\bar{d}_{L,i})$.
\end{theorem}
\begin{proof}
  Since $x_i\in\mc{X}_i(\bar{d}_{L,i},w_i)$, there exists $\beta_i>0$ such that \[\begin{split}\Omega_{i,\beta_i}(\bar{d}_{L,i}) = \{y\in\mc{X}_i^N(\bar{d}_{L,i},w_i)\colon & V_{N,i}^0(y,\bar{d}_{L,i},w_i)\leq\beta_i,\\
      & w_i\in\mbb{W}_i\}\end{split}\]is non-empty and $\Omega_{i,\beta_i}(\bar{d}_{L,i})\subseteq\mc{Z}_i^N(\bar{d}_{L,i})$. From Corollary~\ref{cor:value_function_Lyapunov}, the value function and the state at $t+\delta$ satisfy $x_i^+=x_i(t+\delta)\in\mc{X}_N(\bar{d}_{L,i},w_i)$ and\[V_{N,i}^0(x_i^+,\bar{d}_{L,i},w_i)\leq V_{N,i}^0(x_i,\bar{d}_{L,i},w_i)-\alpha_{i3}(|x_i-x_i^\ts{eq}|).\]On the other hand, Lemma~\ref{lem:value_function_kappa} establishes the $\mc{K}-$continuity such that
  \[V_{N,i}^0(x_i^+,\bar{d}_{L,i},w_i^+) - V_{N,i}^0(x_i^+,\bar{d}_{L,i},w_i) \leq \sigma_V(|w_i^+-w_i|).\]Combining both of these inequalities,  Assumption~\ref{assum:bounded_disturbance}, and $x_i\in\Omega_{i,\beta_i}(\bar{d}_{L,i},w_i)$ yield
  \[\begin{split}
      V_{N,i}^0(x_i^+,\bar{d}_{L,i},w_i^+)  \leq & V_{N,i}^0(x_i,\bar{d}_{L,i},w_i)  - \alpha_{i3}(|x^\ts{eq}-x_i^\ts{eq}) \\ & + \sigma_V(|w_i^+-w_i|)\\
     \leq  & (i_d - \alpha_{i3}\circ\alpha_{i2}^{-1})(\beta_i) + \alpha_{i3}\circ\alpha_{i2}^{-1}(\beta_i)\\
     \leq  & \beta_i.
   \end{split}\]This implies that $x_i^+\in\Omega_{i,\beta}(\bar{d}_{L,i})$, and as a result the set $\Omega_{i,\beta_i}(\bar{d}_{L,i})$ is invariant for the dynamics\[\begin{split}
     x^+_i =& F_i(x_i,\kappa_{N,i}(x_i,w_i),\bar{d}_{L,i}) + E_iw_i\\
     w_i^+ \in & w_i + \Delta\mbb{W}_i.
   \end{split}\]where $\kappa_{N,i}(x_i,w_i)\in\mbb{U}_i$ is the optimal control related to $(x_i,w_i)$. This fact has as a consequence $x_i^+\in\mc{X}_i^N(\bar{d}_{L,i},w_i^+)$. 
\end{proof}
\subsubsection{Closed-loop stability}
\label{sec:clos-loop-stab-1}
In the previous section, we have shown that the controller of each node $i$ is recursively feasible. The proof relies on the stability properties of the MPC controller attached to each node $i$, in particular the value function $i\in\mc{V}$ behaves as a Lyapunov function which depends on the interconnection. In this section, we analyse the stability of the network and construct a Lyapunov function for the complete network. In the analysis, we put emphasis on bounding the interaction disturbance and on guaranteeing that the map characterising the steady state for each node converges to the network equilibrium.

We first analyse the relation between the equilibrium pairs for local systems $(x_i^\ts{eq}(\bar{d}_{L,i},w_i),u_i^\ts{eq}(\mb{d}_{L,i},w_i))$ and those arising from the network equilibrium $(x^\ts{eq},u^\ts{eq})$. The local equilibria, characterised by $F_i(x_i,u_i,\bar{d}_{L,i}) =0$ and  $u_i\in \frac{1}{2}[-I_i^\ts{max},I_i^\ts{max}]$, satisfy $\tilde{i}_i = u_i$, $\sigma_i = \frac{2}{I_i^\ts{max}}u_i$ for both inverter currents and bounded integrators. This allows us to analyse only the voltage equilibria via the following steady state optimisation problem
\begin{equation}
  \begin{split}
    \mbb{P}_i^\ts{ss}(\bar{d}_{L,i},w_i)\colon \min &\{|v_i-v^*|^2\colon   u_i\in \frac{1}{2}[-I_i^\ts{max},I_i^\ts{max}],\\ 
    & v_i \in [V_i^\ts{min},V_i^\ts{max}],\\
     \mc{L}_{ii}&v_i +g_{i}(v_i)\bar{d}_{L,i} = i_{s,i} + u_i +  w_i \}
  \end{split}
  \label{eq:ocp_ss}
\end{equation}
The following property sheds some light on the properties of the steady state optimisations for each node
\begin{proposition}
  If the optimal current $u_i^\ts{ss}\in \frac{1}{2}(-I_i^\ts{max},I_i^\ts{max})$, then the optimal steady state voltage satisfies $v_i^\ts{ss} = v^*$. 
\end{proposition}
\begin{proof}
  The  KKT system for $\mbb{P}_i^\ts{ss}(\bar{d}_i)$ is given by\[\begin{split}
      (v_i-v^*) + \lambda \nabla_v h(v_i,u_i,w_i\bar{d}_{L,i}) + \mu_v^\top\nabla_vr(v_i,u_i) &= 0\\
     \lambda \nabla_u h(v_i,u_i,w_i\bar{d}_{L,i}) + \mu_u^\top\nabla_ur(v_i,u_i)& = 0 \\
     h(v_i,u_i,w_i,\bar{d}_i) & =0\\
     0\leq \mu \bot g(v_i,u_i) & \leq 0 
   \end{split}\]where $h(v_i,u_i,w_i\bar{d}_{L,i})$ defines the equality constraints, $r(v_i,u_i)$ is an affine function characterising the inequality constraints, $(\lambda,\mu_v,\mu_v)\in\Rset^{5}$ are the dual variables, and $a \bot b = a^\top b$ with $a,b\in\Rset^n$. Since by assumption $u_i^\ts{ss}$ is an interior point, then $\mu_u = 0$. This implies that $\lambda = 0$. The reduced KKT system can be expressed as\[\begin{split}
     (v_i-v^*) +  \mu_{v1} - \mu_{v2} &= 0\\
     \min(\mu_{v1},-v_i+V_i^\ts{max})& = 0\\
     \min(\mu_{v2},v_i-V_i^\ts{min})& = 0.\\
   \end{split}\] Since $\min(a,b) =0$ if and only if $(a,b)\geq 0$ or $ab=0$, then from the last equation we obtain two cases: $\mu_{v2} = 0$ or $v_i = V_i^\ts{min}$. The latter condition implies $\mu_{v1} = 0$ and $\mu_{v2} = V_i^\ts{min} - v^*$ which does not satisfy the positivity condition. The first case yields $\mu_{v1} = v^*-v_i$ and two further cases: $\mu_{v1}=0$ or $v_i = V_i^\ts{max}$. The second case leads to a negative multiplier, therefore $\mu_{v1} = 0$ which implies $v_i^\ts{ss} = v^*$. 
\end{proof}
On the other hand, the  network steady state pairs $(x^\ts{eq},u^\ts{eq})$ lie in the manifold given by
\begin{equation}
    \begin{split}
    \mc{H} = \{(v,u)\in\Rset^{|\mc{V}|}\colon &-\mc{L}v + u + i_s - g_L(v)d_L = 0,\\
    & d_L\in\mbb{D},~u\in\mbb{U}\};
  \end{split}
\label{eq:ss_net}
\end{equation}
this manifold can be seen as the level set of a map $\Phi_{{d}_L}\colon\Rset^{|\mc{V}|}\times\Rset^{|\mc{V}|}\to\Rset^{|\mc{V}|}$. The differential of this map is surjective at each point of an open neighbourhood $\mc{N}\supset\mbb{X}\times\mbb{U}$, as a consequence of Assumption~\ref{assum:load_structure} and because of the implicit function theorem there is a map $x^\ts{eq} = \xi(u)$ such that $\Phi_{{d}_L}(\xi(u),u) =0$ for all $u\in\mbb{U}$. This implies that for each current reference, there exists at least an associated voltage. A natural step is to investigate the deviation of $u^\ts{eq}$ and $u^\ts{ss} = (u_1^{ss}(w_1,\bar{d}_{L,1}),\ldots,u_{\mc V}^{ss}(w_{\mc V},\bar{d}_{L,|\mc V|}))$. The solution of $\mbb{P}_i(\bar{d}_{L,i},w_i)$ yields $u_i^\ts{ss}(\bar{d}_{L,i},w_i) = \mc{L}_{ii}v^* + g_i(v^*)\bar{d}_{L,i} -w_i-i_{s,i}$ where the interaction is $w_i = \mc{L}_i^Cv$ where $\mc{L}^C = \mc{L} - \mc{L}^D$ with $\mc{L}^D$ is a diagonal matrix with $\mc{L}_{ii}^D = \mc{L}_{ii}$ for all $i\in\mc{V}$. Following~\eqref{eq:ss_net}, the desired difference is\[\begin{split}
    |u^\ts{ss}-u^\ts{eq}| = & |\mc{L}^D(v^\ts{eq}-v^*\mathds{1}_{|\mc{V}|}) + g_L(v^\ts{eq}) - g_L(\mathds{1}_{|\mc{V}|})\bar{d}_L\\
    & + g_L(v^\ts{eq}(d_L-\bar{d}_L))|\\
    |u^\ts{ss}-u^\ts{eq}| \leq & (|\mc{L}^D| + G)|v^\ts{eq} - v^*\mathds{1}_{|\mc{V}|}| + |g_L(v^\ts{eq})|\sum_{i\in\mc{V}}\gamma_{L,i},
  \end{split}\]where $G = \sum_{i\in\mc{V}}\sum_{j = 1}^{n_{L,i}}\bar{d}_{L,ij}\pdiff{g_{L,ij}}{v_i}\big\vert_{v^*}$ is the Lipschitz constant for the load characteristics, and $\gamma_{L,i}>$ from Assumption~\ref{assum:preview_information}. Following the formulation for each \ac{OCP}, the control law aims to minimise the deviation from $u_i^\ts{ss}$ such that $\kappa_{N,i}(x_i,\bar{d}_{L,i},w_i) = u_i^\ts{ss}(\bar{d}_{L,i},w_i) + \tilde{u}_i$ with $\tilde{u}\in\{f\in\mbb{U}\colon u_i^\ts{ss}(\bar{d}_{L,i},w_i)+f\in\mbb{U}_i\}$. The resulting closed-loop network voltage dynamics are
\begin{equation}
  \begin{split}
    C\dot{v} = & -\mc{L}^D(v-v^*\mathds{1}_{|\mc{V}|}) - (g_L(v) - g_L(v^*\mathds{1}_{|\mc{V}|}))\bar{d}_L\\
    & + \tilde{r}(\tilde{i}-u_i) - g_L(v)(d_L-\bar{d}_L).
  \end{split}
  \label{eq:dyn_cl_net}
\end{equation}
In the above equation, we have tacitly assumed that all node currents $\tilde{i}\to u_i$ and the quantity $\tilde{r}$ is a function of the error between current and reference. The speed of this convergence is governed by~\eqref{eq:Lyap_decrease} and the proportional gain $k_{P,i}$ from the current controller. Further modifications to~\eqref{eq:dyn_cl_net} yield\[\begin{split}
    C\dot{v} = & -(\mc{L}^D+ G)(v-v^*\mathds{1}_{|\mc{V}|}) + \tilde{u}_i\\
    & + \tilde{r}(\tilde{i}-u_i) + \tilde{s}(v-v^*\mathds{1}_{|\mc{V}|}) - g_L(v)(d_L-\bar{d}_L)
  \end{split}\]where we have used the definition of derivative for $g_L(\cdot)$ which implies that $\lim_{v\to v^*\mathds{1}_{|\mc V|}}\frac{|s(v- v^*\mathds{1}_{|\mc V|})|}{|v- v^*\mathds{1}_{|\mc V|}|} = 0$ and $G$ is the derivative of $g_L(\cdot)\bar{d}_L$ evaluated at the reference voltage. Depending on the load characteristics the matrix G can be either negative or positive definite; in the latter case each control input $\tilde{u}_i$ counteracts the effects of both instability of $-(\mc{L}^D +G)$ and the effects of $\tilde{r}$ and $\tilde{s}$. We note that when using $u^\ts{ss}$, the resulting equilibrium manifold can be described as\[\mc{L}^D(v-\mathds{1}_{|\mc V|}) + (g_L(v)-g_L(v^*\mathds{1}_{\mc V}))\bar{d}_L = g_L(v)(d_L-\bar{d}_L),\]the solution of which is clearly $v^\ts{eq} = v^*\mathds{1}_{\mc V}$ when $\bar{d}_L-d_L = 0$, \ie when the controller has perfect knowledge of the load.
\begin{example}[equilibrium for CPLs]
For $g_{L,i}(v) = v^{-1}$, \ie a CPL, the equilibrium manifold, assuming perfect knowledge of the load, has two solutions at $v_i = v_i^*$ and $v_i = \frac{d_{L,i}}{\mc{L}_{ii}v^*}$. The second solution depends on the load demand which may be a natural choice of equilibrium, however for an operating voltage of $560 \si{\volt}$ and an admittance in the order of $o(10^1)$ yields that $v_i \approx 10^{-4}d_{L,i}$ which may be outside the voltage constraint set. When the load is uncertain, the solution, following a perturbation analysis, is $v_i = v^*\mathds{1}_{\mc V} + (d_{L,i}-\bar{d}_{L,i})\frac{v^*}{\bar{d}_{L,i}-\mc{L}_{ii}v^*} + (d_{L,i}-\bar{d}_{L,i})^2\bigl(\frac{v^*}{\bar{d}_{L,i}-\mc{L}_{ii}v^*}\bigr)^3+\cdots$. This implies that $|v^\ts{eq}_i-v^*| \leq \frac{(d_{L,i}-\bar{d}_{L,i})}{\mc{L}_{ii}v^*}$ approximately.
\end{example}
Inspired by the above example, we can attempt to bound the equilibrium deviation from the operating value as a function of the load uncertainty. A similar perturbation analysis yields  $v^\ts{eq} = v^* + |d_L-\bar{d}_L|\nu_1(v^*,\bar{d}_L) + |d_L-\bar{d}_L|^2\nu_2(v^*,\bar{d}_L) + \cdots$ where $\nu_1$, $\nu_2$, etc are functions of both the operating conditions and nominal load. This results in a bound on $|v^\ts{eq}-v^*\mathds{1}_{\mc V}| \leq \varepsilon |\sum_{k=1}^{\infty}\nu_k\varepsilon^{k-1}|$ with $\varepsilon = \sum_{i\in\mc{V}}\gamma_{L,i}$. The following result establishes the properties of the closed loop system around a neighbourhood of the nominal operating point.
\begin{proposition}[Closed loop control invariance near the equilibrium]
  Suppose Assumptions~\ref{assum:load_structure}, \ref{assump:constraints}, and \ref{assum:preview_information} hold. There exists a set $\mbb{S}\subset\Rset^{|\mc V|}$ that is control invariant for the voltage dynamics $\dot{z} = A z + \tilde{u} + \tilde{s}(z) + \tilde{r}(\tilde{i}-u_i) - g_L(z+v^*\mathds{1}_{|\mc V|})(d-\bar{d})$ and constraint sets $(\mbb{V},\mbb{U})$. 
\end{proposition}
\begin{proof}
  The proof is a consequence of the formulation of the \ac{OCP} together with the boundedness of the deviations of voltages, currents, and loads w.r.t $v^*\mathds{1}_{|\mc{V}|}$, $\kappa_{N}(v,\bar{d})$ and $d_L$ respectively.
\end{proof}
The importance of the above proposition is that it ensures the existence of a control action that counteracts both the potential instability introduced by the loads. As time increases, the deviation of current vanishes following \eqref{eq:Lyap_decrease}, and the only remaining terms have potentially constant values of $v^\ts{eq} - v^*\mathds{1}_{|\mc V|}$ and $d_L-\bar{d}_L$. This yields an equilibrium pair $(v^\ts{eq},\tilde{u}^\ts{eq}+\kappa_{N}(v,\bar{d}))\in\mc H$.

The final part of the puzzle is the analysis of the interconnection disturbance. The behaviour of between samples can be characterised by the variation $\Delta w_i = w_i(\delta+t)-w_i(t)$ which can be bounded for all $i\in\mc{V}$ as\[\begin{split}
    |\Delta w_i| \leq & |\mc{L}_i^C| |v^+ - v|)
  \end{split}\]
where $\mc{L}_i^C = [\mc{L}_{ij}]_{i\neq j}\in\Rset^{|\mc{V}|}$ collects all interconnection information for node $i\in\mc{V}$. On the other hand, the difference between the solutions of $\mbb{P}_i^\ts{ss}(w_i(t+\delta),\bar{d}_{L,i})$ and $\mbb{P}_i^\ts{ss}(w_i(t),\bar{d}_{L,i})$ behaves in a similar way, \ie \[\begin{split}
    |u_i^\ts{ss}(w_i(t+\delta),\bar{d}_{L,i}) - u_i^\ts{ss}(w_i,\bar{d}_{L,i})| \leq & |\mc{L}_i^C| |v^+ - v|.
  \end{split}\]
We note that the difference $|v^+ - v|$ shrinks as the voltage approaches its equilibrium. Moreover, from the closed loop dynamics~\eqref{eq:dyn_cl_net}, the bounds for this  difference is\[\begin{split}
    |v^+-v|\leq & \biggl|\int_t^{t+\delta}-\mc{L}^D(v-v^*\mathds{1}_{\mc V}) - (g_L(v) - g_L(v^*\mathds{1}_{\mc V}))\bar{d}_L\\
    &\quad + g_L(v)(d_L-\bar{d}_L) + \tilde{u} dt\biggr|\\
    \leq &(G+|\mc{L}^D|)\int_{t}^{t+\delta}|v-v^*\mathds{1}_{\mc V}| dt + \delta|u| \\
    &\quad+ |g_L(v)|\sum_{i\in\mc V}\gamma_{L,i}\\
    & \alpha_V(|v-v^*\mathds{1}_{\mc V}|) + \delta (|\tilde{u}| + \Gamma\sum_{i\in\mc V}\gamma_{L,i})
  \end{split}\]
In the last inequality, we have used the following well known property of $\mc{K}-$functions: the integral of a class $\mc{K}-$function is also class $\mc{K}$. Given the continuity of the load characteristics, we have $|g_L(v)| \leq \Gamma$ for all $v$, and in between samples the action of the secondary controller is considered constant. We are now in position to state the main result of this paper:
\begin{theorem}[Closed-loop stability]
  Suppose Assumption~\ref{assum:load_structure},\ref{assump:constraints}--\ref{assum:preview_information} hold. If in addition, $\sigma_V(r) = L_Vr$ for $r>0$ and $L_V>0$, and there exists $\eta_i>0$ for all $i\in\mc{V}$ such that \[\alpha(r) = \sum_{i\in\mc V}\eta_i\alpha_{3i}(r) - L_V(\sum_{i\in\mc{V}}\eta_i|\mc{L}_i^C|)\alpha_{V}(r)\]is a $\mc{K}-$function, then the network of buck converters~\eqref{buck} is input to state stable (ISS) in closed loop with the control law~\eqref{eq:control_law} and current references given by $\tilde{i}_i^\ts{ref} = \kappa_{N,i}(x_i,w_i,\bar{d}_{L,i})$. 
\end{theorem}
\begin{proof}
  Following Corollary~\ref{cor:value_function_Lyapunov}, the value function for each $i\in\mc{N}$ is a Lyapunov function with respect to an equilibrium point $x^\ts{eq}_i(\bar{d}_{L,i},w_i)\in\mbb{X}_i$ which is dependent on both nominal load and interaction with the network. Furthermore, Proposition~\ref{prop:lyapunov_driving} provides us with a Lyapunov function for the driving subsystem for each $i\in\mc{V}$. Our strategy to prove the asymptotic stability of the closed-loop system hinges on showing that
  \begin{equation}
    \Psi(x) = \sum_{i\in\mc{V}}\eta_i(\underbrace{V_{N,i}^0(x_i,\bar{d}_{L,i},w_i) + W_i(x_i)}_{\Psi_i(x_i,\bar{d}_{L,i},w_i)})
    \label{eq:network_lyap}
  \end{equation}
  is a Lyapunov function. Here, $V_{N,i}^0(\cdot,\cdot)$ is the \ac{OCP} value function and $W_i(\cdot)$ is the bounded integrator Lyapunov function defined in \eqref{eq:current_Lyapunov}; the constants $\eta_i>0$ are suitable weights as in \cite{Siljak2007}. The variation $\Delta\Psi(x) = \Psi(x(t+\delta)) - \Psi(x(t))$, following Corollary~\ref{cor:value_function_Lyapunov} and Proposition~\ref{prop:lyapunov_driving}, is
   \[\begin{split}
    \Delta \Psi(x) \leq &\sum_{i\in\mc{V}}\eta_i \bigl (-\alpha_{3i}(|x_i-x^*|) + \sigma_V(|\Delta w_i)|) \\
    &- \gamma_i(r + k_{P,i})\int_t^{\delta +t}(\tilde{i}_i-\kappa_{N,i}(x_i,w_i,\bar{d}_{L,i}))^2dt\bigr ) \\
    \leq & \sum_{i\in\mc{V}}\eta_i \bigl (-\frac{1}{3}\alpha_{3i}(|v_i-v_i^\ts{eq}|) \\
    & - \tilde{\alpha}_{4i}(|\tilde{i}_i-\kappa_{N,i}(x_i,w_i,\bar{d}_{L,i})|)\\
    & -\frac{1}{3}\alpha_{3i}(|\sigma_i- \frac{2}{I_i^\ts{max}}\kappa_{N,i}(x_i,w_i,\bar{d}_{L,i})|)\bigr)\\
    & + L_V\sum_{i\in\mc{V}}\eta_i|\mc{L}_i^C|(\alpha_V(|v-v^*\mathds{1}_{\mc V}|) \\
    & + \delta(|\tilde{u}| + \Gamma\sum_{i\in \mc V}\gamma_{L,i})))\\
    \leq &  -\alpha(|v-v^*\mathds{1}_{\mc V}|) - \tilde{\alpha}(|\tilde{i}-\kappa_{N}(x,\bar{d}_{L,i})|)\\
    & -\hat{\alpha}(|\sigma- \frac{2}{I^\ts{max}}\kappa_{N}(x,\bar{d}_{L,i})|) \\
    & + L_V\sum_{i\in\mc{V}}\eta_i|\mc{L}_i^C|\delta(|\tilde{u}| + \Gamma\sum_{i\in \mc V}\gamma_{L,i})). 
  \end{split}\]
In the above chain of inequalities, we have used the hypothesis that claims $\alpha(\cdot)$ is a $\mc{K}-$function. Furthermore $\tilde{\alpha} = \sum_{i\in\mc{i}}\eta_i\tilde{\alpha}_{4i}$ where $\tilde{\alpha}_{4i}r = \frac{1}{3}\alpha_{3i}(r) + \int_t^{t+\delta}|r|dt$ and $\hat{\alpha} = \sum_{i\in\mc{V}}\frac{\eta_i}{3}\alpha_{3i}$ are $\mc{K}$-functions. Following that $|\tilde{u}|\to0$ as the system approaches its equilibrium, and the uncertainty of the load is bounded by Assumption~\ref{assum:preview_information}, the ISS of the buck converter network follows.
\end{proof}
\section{Simulations}
\label{sec:simulations}
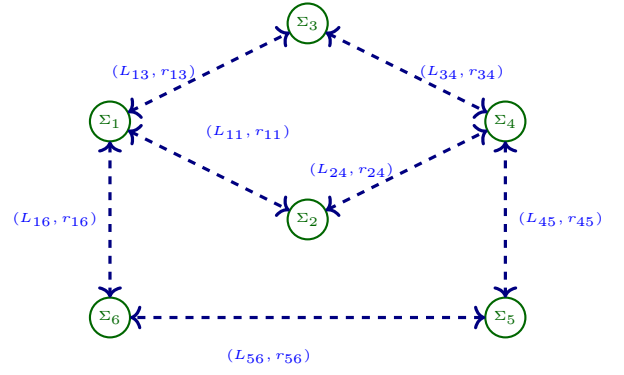
\begin{figure}[t!]
  \centering
  \begin{tikzpicture}[scale=0.65,font=\tiny]
  \begin{scope}[anchor=west,every node/.style={circle,thick,draw,minimum size=15,inner sep=0pt, outer sep=0pt},green!40!black]
    \node (S1) at (-4,2) {$\Sigma_1$};
    \node (S2) at (0,0) {$\Sigma_2$};    
    \node (S3) at (0,4) {$\Sigma_3$};
    \node (S4) at (4,2) {$\Sigma_4$};
    \node (S5) at (4,-2) {$\Sigma_{5}$};
    \node (S6) at (-4,-2) {$\Sigma_6$};
  \end{scope}
\begin{scope}[blue,decoration={ markings,
    mark=at position 0.5 with {\arrow{>}}},every node/.style={circle}, every edge/.style={draw=blue!50!black,very thick,dashed}]
  \path[<->] (S1) edge node[above right] {$(L_{11},r_{11})$} (S2);
  \path[<->] (S1) edge node[left] {$(L_{13},r_{13})$} (S3);
  \path[<->] (S2) edge node[left] {$(L_{24},r_{24})$} (S4);
  \path[<->] (S3) edge node[right] {$(L_{34},r_{34})$} (S4);
  \path[<->] (S1) edge node[left] {$(L_{16},r_{16})$} (S6);
  \path[<->] (S4) edge node[right] {$(L_{45},r_{45})$} (S5);
  \path[<->] (S5) edge node[below left] {$(L_{56},r_{56})$} (S6);
\end{scope}
\end{tikzpicture}
  \caption{Network topology: each node $\Sigma_i$ is interconnected with a subset $\mc{V}_i$ of $\{1,\ldots,6\}$ via lines characterised by parameters $(L_{ij},r_{ij})$. Each node is assumed to be feeding a load.}
  \label{fig:topology}
\end{figure}
In this section, we explore the behaviour of a network of $|\mc{V}| = 6$ buck converters, see Fig.~\ref{fig:topology} for the interconnection topology, where each source is feeding a constant power load, \ie each $g_{L,i}(v_i) = v_i^{-1}$. The dynamics for each power converter are:\[\begin{split}
    C_i\frac{dv_i}{dt} &= -\mc{L}_{ii}v_i -\frac{P_i}{v_i} + i_i - \sum_{j\in\mc{V}_i}\mc{L}_{ij}v_j\\
    L_i\frac{di_i}{dt} & = -r_ii_i + V_\ts{in}u_i - v_i.
  \end{split}\]
The input voltage for each controller is $V_{in} = 800 \si{\volt}$, and given that all sources belong to the same connected component of the network, the operating voltage is set to $v^* = 560\si{\volt}$. The rated power for each converter is given as $P_{C,i}^\ts{max} =\{43,39,46,39,50,42\}\si{\kilo\watt}$. We note that the load characteristic function is not defined around the origin and each load draws more current as the voltage plummets. This observation lead us to define the voltage constraint set $\mbb{V}_i = V_\ts{in}[0.3,1]$ which yield the upper bounds for the currents:  $I_i^\ts{max} = \frac{P_{C,i}^\ts{max}}{0.3V_\ts{in}}$ such that $I^\ts{max} = \{178.7,160.9, 193.2,162.1,207.9,173.2\}\si{\ampere}$. The primary controller is given by~\eqref{eq:eq_current} and the translation $\tilde{i} = i_i - i_{s}$ with $i_s = \frac{1}{2}I_i^\ts{max}$. The cost used in the secondary controller for each $i\in\mc V$ is $\ell_i(x_i,u_i) = q_i|v_i-v^*|^2 + n_i|u_i-u_i^\ts{ss}(w_i,\bar{d}_i)|$. The terminal set is $\mbb{X}_i^f = v^*+[-\Delta V_i,\Delta V_i]\times\frac{1}{2}[-I_i^\ts{max},I_i^\ts{max}]\times[-1,1]$ with $\Delta V_i = 10\si{\volt}$.


The simulation shows the behaviour of the network to \emph{a priori} unknown load changes; these changes may occur as a step function in the load value or as uncertainty around each nominal value. Each power converter feeds a load, but in this example we have some loads exceeding the rated capacity of their corresponding power converter, albeit the condition $\sum_{i\in\mc V} P_{L,i} \leq \sum_{i\in\mc V}P_{C,i}^\ts{max}$ always holds. When power converter $i$ cannot feed its own load, the rest of the network aids this converter while accounting for losses in the network. The power demanded by each load is given in Fig.~\ref{fig:power} together with the power provided by the network. The load steps in the following way: $P_{L,1} = 0.95P_{C,i}^\ts{max}$ and at $t= 0.6\si{\second}$ it switches to $P_{L,1} = 0.735P_{C,1}^\ts{max}$; $P_{L,2} = 1.14P_{C,2}^\ts{max}$ and at $t= 1.24\si{\second}$ it switches to $P_{L,2} = 0.73P_{C,2}^\ts{max}$; $P_{L,4} = 0.5P_{C,4}^\ts{max}$ and at $t= 0.93\si{\second}$ it switches to $P_{L,4} = 1.03P_{C,4}^\ts{max}$; and $P_{L,6} = 0.66P_{C,6}^\ts{max}$ and at $t= 0.3\si{\second}$ it switches to $P_{L,6} = 1.05P_{C,6}^\ts{max}$. In Fig.~\ref{fig:power}, we see how node $i=3$ changes its power, despite its load retaining its nominal value, to account for a large step in node $i=4$. As seen in Fig.~\ref{fig:voltage}, the voltage of each node reacts to a change in load without leaving the terminal set, we note the non-minimal phase type behaviour exhibited by all the voltages when reacting to a load step.  The currents injected $i_i = i_{s,i} + \tilde{i}_i$ by each converter can be seen in Fig.~\ref{fig:current}. We note how these currents do not operate close to their operating limits even when the load exceeds the capacity of a converter; the network itself provides required support for each converter.  
\begin{figure}[t!]
  \input{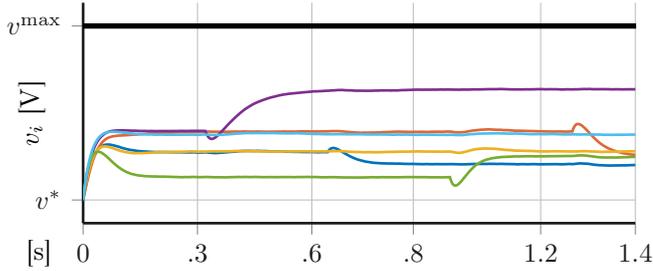}
  \caption{Voltage responses of the network to arbitrary and uncertain load variations. All voltages remain within a compact neighbourhood of the operating voltage $v^*$.}
  \label{fig:voltage}
\end{figure}
\begin{figure}[t!]
  \input{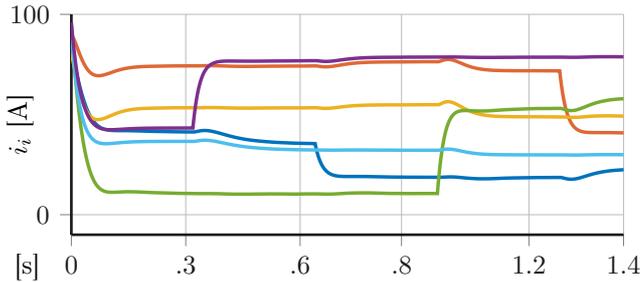}
  \caption{Current responses to load variations. When a drastic load change occurs, see for example \eqref{fig:v:v_5}, the other nodes respond accordingly to maintain balance in the network }
  \label{fig:current}
\end{figure}
\begin{figure}[t!]
  \input{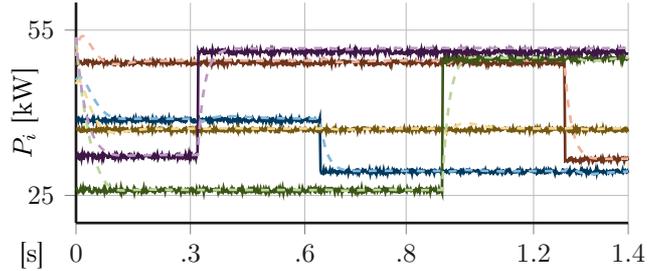}
  \caption{Power response of the network. Solid lines, \eg \eqref{fig:p:p_1}, represent power drawn by each CPL and dashed lines, \eg  \eqref{fig:p:d_1}, represent power provided by each converter.}
  \label{fig:power}
\end{figure}
We focus now in the analysis of the equilibrium points of the network. From Fig.~\ref{fig:voltage} and Fig.~\ref{fig:current}, both voltages and currents clearly converge to neighbourhoods of the equilibrium points. In Fig.~\ref{fig:current_ivsiss}, we observe the deviation of converter currents with respect to the steady state value computed in \eqref{eq:ocp_ss} for each $i\in\mc{V}$. The deviation from this nominal equilibrium point shrinks while the nominal load does not step. The difference, however, does not converge to zero because of the uncertainty in the real load which results in a deviation in terms of an equilibrium point. We see the effect of this deviation is Fig.~\ref{fig:voltage_ss} where we portray the distance of the centralised equilibrium $v^\ts{eq}$ computed for the uncertain load to the operating point $v^*\mathds{1}_{|\mc V|}$ and the voltage $v$. As shown in previous sections, the difference $|v-v^\ts{eq}|$ is bounded by a quantity that depends on the uncertainty in the load. Lastly, we portray in Fig.~\ref{fig:kernel_distance} the distance of the voltage to the kernel of the network Laplacian. As shown in the theory above, this distance does not shrink to zero, but remains in a neighbourhood of $\ker\mc{L}$. The changes in this distance are because of load steps, but despite the changes this distance remains bounded.
\begin{figure}[t!]
  \input{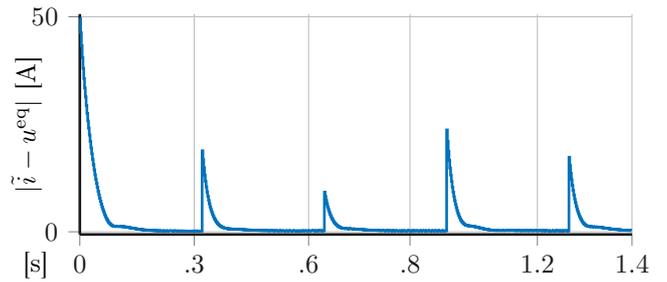}
  \caption{Deviation from current equilibrium $|\tilde{i}-u^\ts{eq}|$. The distance to the equilibrium shrinks while the nominal load is kept constant.}
  \label{fig:current_ivsiss}
\end{figure}
\begin{figure}[t!]
  \input{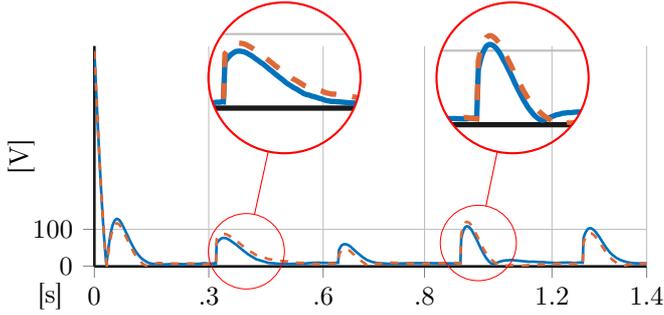}
  \caption{Deviation from voltage centralised equilibrium, \eqref{fig:vss:veq-vs} represents $|v^\ts{eq}-v^*\mathds{1}_{|\mc V|}|$ and \eqref{fig:vss:veq-v} represents $|v^\ts{eq}-v|$. The distance from the equilibrium voltage $v^\ts{eq}$ decreases in time when the nominal does not change.}
  \label{fig:voltage_ss}
\end{figure}
\begin{figure}[t!]
  \input{figs/dker.tikz}
  \caption{Distance between $v\in\Rset^{|\mc V|}$ and $\ker\mc L$.}
  \label{fig:kernel_distance}
\end{figure}

\section{Conclusions}
\label{sec:conclusions}
In this paper, we have proposed a decentralised primary controller, together with a distributed secondary controller for a DC network of Buck power converters connected in meshed topology. For the primary controller, we have proposed a Lyapunov function candidate defined over a compact set which is considered a safety region, \ie a set containing all currents below a given threshold. Using this Lyapunov function, we have proven the asymptotic stability of an equilibrium point using the invariance principle. We have analysed the interconnection between voltages and currents to conclude the attractiveness of a neighbourhood of the kernel of the network Laplacian matrix. We leverage on this result to prove the ISS behaviour of the closed loop system between the network and a distributed receding horizon voltage controller. We have also proven recursive feasibility of the controller. Lastly, we illustrated our approach in a network containing uncertain constant power loads with \emph{a-priori} unknown variation. The results show how our proposed controller is capable to steer the voltages near a nominal operating voltage; furthermore, the simulation results show how different members of the network are able to cooperate when a load exceeds the capacity of its hosting power converter.
\bibliography{extracted}%
\end{document}